\documentclass[aps,twocolumn,superscriptaddress,footinbib]{revtex4}
\usepackage{epstopdf}
\usepackage{graphicx}  
\usepackage{bm}        
\usepackage{braket}
\usepackage{exscale}                  
\usepackage[intlimits]{amsmath}       
\usepackage{amsfonts}
\usepackage{amssymb,amscd}
\usepackage{color}
\usepackage{hyperref}
\usepackage{subfigure}
\usepackage[normalem]{ulem}
\usepackage{slashed}
\usepackage{leftidx}
\usepackage{xcolor,soul}
\usepackage{txfonts}
\usepackage{charter}
\usepackage{dsfont}
\usepackage{mathrsfs}

\newtheorem{lemma}{Lemma}

\newtheorem{proposition}{Proposition}
\newtheorem{definition}{Definition}

\def\Tr{\operatorname{tr}}
\usepackage{times}

\hypersetup{colorlinks=true, citecolor=blue, urlcolor=blue, linkcolor=blue}
\begin{document}
\title{Geometric Power Capacity of Coherent Ergotropy in Quantum Batteries}

\author{Dong-Ping Xuan\footnote{E-mail: xuandongping@ccut.edu.cn}}
\affiliation{School of Mathematics and Statistics, Changchun University of Technology, Changchun 130012, Jilin, China}
\author{Hua Nan\footnote{E-mail: nanhua@ybu.edu.cn}}
\affiliation{Department of Mathematics, College of Sciences, Yanbian University, Yanji 133002, China}
\author{Zhi-Xi Wang\footnote{E-mail: wangzhx@cnu.edu.cn}}
\affiliation{School of Mathematical Sciences, Capital Normal University, Beijing 100048, China}
\author{Shao-Ming Fei\footnote{E-mail: feishm@cnu.edu.cn}}
\affiliation{School of Mathematical Sciences, Capital Normal University, Beijing 100048, China}
\affiliation{Max-Planck-Institute for Mathematics in the Sciences, 04103 Leipzig, Germany}
\begin{abstract}
\textbf{We explore coherent ergotropy extraction in quantum batteries from a resource-geometric point of view. 
For an initial state $\rho$, we quantify the coherent extraction process by the coherent ergotropy 
$\mathcal{E}_c(\rho)$ and the coherent extraction distance $D_c^{\rm ext}(\rho)$ between the active state $\sigma_\rho$ and the passive state $P_\rho$. This defines the geometric power capacity $\Pi_c(\rho)=\mathcal{E}_c(\rho)/D_c^{\rm ext}(\rho)$, which measures the coherent ergotropy 
released  unit minimal unitary distance. We prove that, for any driving Hamiltonian satisfying 
$\|V_t\|\leq\nu$, the actual coherent discharging power is bounded by $P_c^{\rm ext}(\rho;V_t)\leq \nu\Pi_c(\rho)$,
showing that $\Pi_c(\rho)$ is a capacity under unit driving norm rather than the power of a particular protocol.
General bounds on $\Pi_c(\rho)$ are derived by combining relative entropy bounds on coherent 
ergotropy with geometric bounds on the coherent extraction distance. 
We also formulate coherence measure induced bounds and protocol-corrected capacities 
involving the effective speed of a given Hamiltonian. Qubit and qutrit examples demonstrate 
that $\Pi_c(\rho)$ captures a resource-geometric feature of coherent discharging beyond coherent 
ergotropy or coherence measures alone.}
\end{abstract}
\maketitle

\vspace{.4cm}

\section{Introduction}\label{sect1}

Quantum batteries were first introduced by Alicki and Fannes as quantum systems designed for storing and releasing energy at the microscopic scale~\cite{Alicki2013PRE}. In contrast to classical batteries, quantum batteries can exploit genuinely quantum features, such as coherence, entanglement, collective effects, and nonclassical correlations, which, under suitable physical and dynamical conditions, may improve their charging power, energy-storage capability, and stability~\cite{Andolina2019PRB,Andolina2019PRL}. However, the presence or amount of entanglement alone does not guarantee a charging-power advantage. In particular, the quantum-state contribution to charging power is not generally an entanglement monotone, and entanglement may be either beneficial or detrimental depending on the battery state and the charging protocol~\cite{GyhmFischer2024}.
Over the past decades, quantum batteries have attracted considerable attention, and a variety of theoretical models and protocols have been proposed to clarify when and through which mechanisms quantum resources
influence their performance. In particular, many studies have investigated optimal charging and
discharging processes aimed at improving the power performance of quantum batteries
~\cite{Ferraro2018,Binder2015,Le2018,Andolina2018,Ghosh2020,GarciaPintos2020}. Other studies have focused on work storage, energy extraction, and the stability of stored energy under realistic dynamical conditions
~\cite{Crescente2020,Gherardini2020,Santos2019,Quach2020}. More recently, topological photonic waveguides and their associated bound states have been employed to achieve near-perfect energy transfer, protect ergotropy against dissipation under suitable configurations, and temporarily enhance charging power through the quantum Zeno effect~\cite{Lu2025PRL}. These developments show that the performance of a quantum battery is not determined solely by the amount of energy stored in the system, but also by the physical mechanisms through which energy is charged, protected, and extracted.


The energetic content of a quantum battery is commonly characterized by its ergotropy, which quantifies the maximum amount of work that can be extracted from a quantum state by means of cyclic unitary operations~\cite{Alicki2013PRE,ergotropyAllahverdyan2004Europhys}. For a quantum system with Hamiltonian \(H\) and state \(\rho\) acting on a \(d\)-dimensional Hilbert space, the ergotropy is defined by 
$$
\mathcal{E}(\rho)=\operatorname{Tr}(\rho H)-\min_{U}\operatorname{Tr}(U\rho U^\dagger H),
$$
where the minimization is taken over all unitary transformations.
Equivalently, ergotropy measures the energy difference between the initial state and its passive state, namely, the state from which no further work can be extracted under unitary operations~\cite{Ferraro2018,Andolina2018,Ghosh2020,GarciaPintos2020,Farina2019,Barra2019PRL}.
As such, ergotropy provides a natural figure of merit for the charge level of a quantum battery and has been widely used to analyze work storage, work extraction, charging advantages and discharging protocols in quantum thermodynamic systems
~\cite{Barra2022NJP,SongWLPRL,Campaioli2024RMP,Monsel2020PRL,Tirone2021PRL,GelbwaserKlimovsky2013,Niedenzu2016NJP,Seah2018NJP,Niedenzu2019Quantum,Seah2021PRL}. 

Physically, it is important to identify which part of ergotropy originates from genuinely quantum coherence. In finite-time thermodynamic transformations, controlled operations generally generate coherence in the instantaneous energy eigenbasis of the driven system, and such coherence can act as a nonequilibrium resource for work extraction. Francica \textit{et al.} introduced the notion of coherent ergotropy to isolate the coherent contribution to the total work yield~\cite{coherentergotropyPRL,BPRE2020}. The optimal work-extraction process can be divided into an incoherent operation and a coherence-extraction cycle, allowing one to distinguish the incoherent and coherent parts of the extractable work. In this sense, coherent ergotropy quantifies the maximum amount of work that can be attributed specifically to quantum coherence, rather than to population imbalance alone. 
This concept provides a useful tool for assessing the contribution of energy-basis coherence to work extraction and coherent discharging performance, without assuming that coherence alone universally determines the attainable power.

A complementary perspective has recently been provided by the concept of quantum charging distance~\cite{Ju-Yeon2024PRA}. Rather than focusing only on the amount of energy or work stored in a quantum battery, this approach characterizes the minimal unitary distance, or equivalently the minimal time under bounded driving resources, required to transform one battery state into another. For pure states, this distance reduces to the Bures angle, whereas for mixed states it generally involves an optimization over admissible unitary transformations and admits computable geometric bounds. This shows that the charging or discharging performance of a quantum battery is constrained not only by the amount of extractable work, but also by the geometry of the state-space path through which such work is accessed.

The above observations naturally motivate the question addressed in this work. Coherent ergotropy quantifies how much extractable work originates from quantum coherence, while quantum charging distance quantifies the geometric cost of implementing a unitary state transformation. However, a state may contain a large amount of coherent ergotropy but require a long coherent extraction path, whereas another state may contain less coherent ergotropy but release it along a shorter path. We therefore combine these two viewpoints and introduce a geometric characterization of coherent ergotropy extraction. In particular, we define the geometric power capacity of coherent ergotropy as the coherent ergotropy extractable in unit minimal coherent extraction distance.

The paper is organized as follows. Sec.~\ref{sect2} reviews the necessary preliminaries.
Sec.~\ref{sect3} introduces the coherent ergotropy extraction distance and the geometric power capacity, with analytical examples.
Sec.~\ref{sect4} derives general bounds on this capacity and Sec.~\ref{sect5} discusses bounds induced by different coherence measures.
Sec.~\ref{sec:protocol_corrected_capacity} connects the intrinsic capacity $\Pi_c(\rho)$ with concrete driving protocols through a protocol-corrected formulation.
Sec.~\ref{sec:conclusion} summarizes the results.

\section{Preliminaries}\label{sect2}
We first recall some basic notions used throughout this work, including quantum batteries, ergotropy, coherent ergotropy and the quantum charging distance. 

\subsection{Quantum battery and ergotropy}
A quantum battery is modeled as a quantum system whose internal energy is determined by a fixed Hamiltonian,
\begin{equation*}
    H_0=\sum_{k=1}^{d}\epsilon_k
    \ket{\epsilon_k}\bra{\epsilon_k},
    \qquad
    \epsilon_1\leq \epsilon_2\leq \cdots \leq \epsilon_d ,
    \label{eq:H0}
\end{equation*}
where $\{\ket{\epsilon_k}\}_{k=1}^{d}$ is the energy eigenbasis. 
The average energy of a battery state $\rho$ is given by $E(\rho)=\Tr[H_0\rho]$.

Work can be extracted from the battery by a cyclic unitary transformation $U$, generated by a time-dependent Hamiltonian whose initial and final values coincide with $H_0$. Under such a process, the state evolves as
$\rho \rightarrow U\rho U^\dagger$, and the average work extracted from the system is 
$$ 
W(\rho,U) =\Tr\!\left[H_0(\rho-U\rho U^\dagger)\right].
$$
The maximum work extractable from $\rho$ by cyclic unitary transformations is called the ergotropy~\cite{Alicki2013PRE,ergotropyAllahverdyan2004Europhys}, 
$$
\mathcal{E}(\rho) =\max_U W(\rho,U).\label{eq:ergotropy_def}
$$

Let the spectral decomposition of $\rho$ be
\begin{equation}
    \rho=\sum_{k=1}^{d} r_k \ket{r_k}\bra{r_k},
    \qquad
    r_1\geq r_2\geq \cdots \geq r_d .
    \label{eq:rho_spectral}
\end{equation}
The passive state associated with $\rho$ is defined as
\begin{equation}
    P_\rho
    =\sum_{k=1}^{d} r_k
    \ket{\epsilon_k}\bra{\epsilon_k}.
    \label{eq:passive_state}
\end{equation}
That is, $P_\rho$ is obtained by placing the eigenvalues of $\rho$ in decreasing order on the energy eigenstates in increasing order of energy. This state is passive in the sense that no work can be further extracted from it by any cyclic unitary operation. Therefore, the ergotropy can be written as
\begin{equation}
\begin{aligned}
\mathcal{E}(\rho) &= \max_{U} W(\rho, U)  = \operatorname{Tr}\{H(\rho - P_\rho)\} \\
&\equiv \sum_k \varepsilon_k (\rho_{kk} - r_k),
\end{aligned}
 \label{eq:ergotropy_passive}
\end{equation}
where  $\rho_{kk} = \sum_{k'} r_{k'} |\langle r_{k'} | \varepsilon_k \rangle|^2$ denotes the population of $\rho$ in the $k$th energy eigenstate.

Thus, ergotropy measures the useful part of the internal energy of a quantum battery, namely, the amount of energy that can be converted into work through cyclic unitary dynamics.
\subsection{Coherent ergotropy}
We now recall the decomposition of ergotropy into incoherent and coherent contributions~\cite{coherentergotropyPRL,BPRE2020}.
The coherence is always understood with respect to the energy eigenbasis
\(\{\ket{\epsilon_k}\}\) of \(H_0\).
The full ergotropy can be decomposed as
\begin{equation}
    \mathcal{E}(\rho)
    =
    \mathcal{E}_i(\rho)
    +
    \mathcal{E}_c(\rho),
    \label{eq:ergotropy_decomposition}
\end{equation}
where \(\mathcal{E}_i(\rho)\) is the incoherent ergotropy and
\(\mathcal{E}_c(\rho)\) is the coherent ergotropy.

The incoherent ergotropy is the maximum amount of work that can be extracted without consuming the coherence of \(\rho\), namely, under incoherent cyclic unitaries which only reshuffle the energy eigenbasis up to phase factors, $ V_\pi =\sum_{k=1}^{d} e^{-i\phi_k} \ket{\epsilon_k}\bra{\epsilon_{\pi_k}}$,
where \(\pi\) is a permutation of the indices.
Such unitaries preserve the amount of coherence of the state in the energy eigenbasis.

Let \(\tilde{\pi}\) be the permutation that rearranges the energy populations of \(\rho\) in decreasing order, namely, 
$ \rho_{\tilde{\pi}_j\tilde{\pi}_j}\geq\rho_{\tilde{\pi}_{j+1}\tilde{\pi}_{j+1}},\quad j=1,\ldots,d-1 $.
Let
$\sigma_{\rho} = V_{\tilde{\pi}} \rho V_{\tilde{\pi}}^{\dagger} = \sum_k \sum_{k'} \rho_{\tilde{\pi}_k, \tilde{\pi}_{k'}} |\varepsilon_k\rangle\langle\varepsilon_{k'}|$. The incoherent contribution to ergotropy is
\begin{equation}
\mathcal{E}_i(\rho) = \operatorname{Tr}\{H(\rho - \sigma_{\rho})\} = \sum_k \varepsilon_k (\rho_{kk} - \rho_{\tilde{\pi}_k \tilde{\pi}_k}).\label{eq:incoherent_ergotropy}
\end{equation} 
The state \(\sigma_\rho\) is obtained from \(\rho\) by an incoherent unitary transformation. Therefore, \(\sigma_\rho\) and \(\rho\) have the same eigenvalues and the same amount of coherence. However, among the states connected to \(\rho\) by incoherent cyclic unitaries, \(\sigma_\rho\) has the lowest average energy.

Equivalently, let \(\Delta\) denote the full dephasing map in the energy eigenbasis of \(H_0\). Then $\delta_\rho=\Delta[\rho]$ has the same energy populations and the same average energy as \(\rho\), but contains no off-diagonal coherence in the energy eigenbasis. Denote by \(P_\delta\) the passive state associated with \(\delta_\rho\), which is obtained by rearranging the populations of \(\delta_\rho\), equivalently the populations \(\{\rho_{kk}\}\) of \(\rho\), in decreasing order and assigning them to the energy eigenstates in increasing order of energy.
The incoherent ergotropy can be equivalently written as the ergotropy of the dephased state,
\begin{equation}
    \mathcal{E}_i(\rho)
    =
    \mathcal{E}(\delta_\rho)
    =
    \Tr\!\left[
    H_0(\delta_\rho-P_\delta)
    \right].
    \label{eq:Ei_dephased}
\end{equation}
Hence, the incoherent ergotropy depends only on the energy populations of \(\rho\), but not on its off-diagonal coherence.

The coherent ergotropy is defined as the remaining part of the full ergotropy after the incoherent contribution has been extracted,
\begin{equation}
    \mathcal{E}_c(\rho)
    =
    \mathcal{E}(\rho)-\mathcal{E}_i(\rho).
    \label{eq:coherent_ergotropy_def}
\end{equation}
Using Eqs.~\eqref{eq:ergotropy_passive} and \eqref{eq:incoherent_ergotropy}, one obtains
\begin{equation}
    \mathcal{E}_c(\rho)
    =
    \Tr\!\left[H_0(\sigma_\rho-P_\rho)\right]\equiv \sum_k \varepsilon_k (\rho_{\tilde{\pi}_k \tilde{\pi}_k} - r_k).
    \label{eq:coherent_ergotropy}
\end{equation}
This equation gives a physical meaning of the coherent ergotropy:
an incoherent unitary extracts \(\mathcal{E}_i(\rho)\) from \(\rho\) and transforms $\rho$ to \(\sigma_\rho\),
then the remaining extractable work, \(\mathcal{E}_c(\rho)\), is obtained by transforming \(\sigma_\rho\) into the passive state \(P_\rho\).

A crucial observation is that \(\sigma_\rho\) and \(P_\rho\) have the same eigenspectrum. Indeed, \(\sigma_\rho\) is unitarily connected to \(\rho\), while \(P_\rho\) is the passive state obtained by rearranging the eigenvalues of \(\rho\).
Therefore, $\operatorname{spec}(\sigma_\rho)= \operatorname{spec}(P_\rho) = \operatorname{spec}(\rho)$.

As a result, the transformation \(\sigma_\rho \rightarrow P_\rho\)
can be realized by a unitary evolution. This transformation is precisely the coherent part of the work extraction process.
If \(\rho\) is diagonal in the energy eigenbasis, then no energetic coherence is present. In this case, $ \rho=\Delta[\rho], \quad \mathcal{E}_c(\rho)=0$.
Therefore, \(\mathcal{E}_c(\rho)\) quantifies the part of extractable work that is genuinely due to the coherence of the initial state.

\subsection{Quantum charging distance}
We next introduce the quantum charging distance~\cite{Ju-Yeon2024PRA}, which measures the minimal time required to transform one state into another by unitary dynamics under a fixed driving resource.

Consider two density matrices $\rho$ and $\sigma$ with the same eigenspectrum. Since unitary evolution preserves the spectrum, there exists a unitary operator $U$ such that
$\sigma=U\rho U^\dagger$. Assume that $ U=\mathcal{T}\exp\!\left( -i\int_0^T V_t\,dt
 \right)$ is generated by a possibly time-dependent driving Hamiltonian $V_t$, where $\mathcal{T}$ denotes time ordering and $T$ denotes the duration of the driving protocol, i.e., the total evolution time over which $V_t$ generates the transformation $\rho\to\sigma$. The quantum charging distance between $\rho$ and $\sigma$ is defined as
\begin{equation}
    D(\rho,\sigma)
    =
    \min_{V_t:\|V_t\|=1} T,
    \label{eq:charging_distance_def}
\end{equation}
where $\|V_t\|$ is the operator norm, i.e., the largest absolute eigenvalue of $V_t$.
The charging distance measures the intrinsic minimal time cost of a unitary transformation when the driving strength is fixed.

The quantity $D(\rho,\sigma)$ is a proper distance on the set of density matrices with the same spectrum. Namely, it satisfies the positivity, $D(\rho,\sigma)\geq 0$ and $D(\rho,\sigma)=0 \Leftrightarrow \rho=\sigma$;
the symmetry, $D(\rho,\sigma)=D(\sigma,\rho)$,
and the triangle inequality,
$D(\rho_1,\rho_2)\leq D(\rho_1,\rho_3)+D(\rho_3,\rho_2)$.

Although Eq.~\eqref{eq:charging_distance_def} is defined by optimizing over time-dependent Hamiltonians, the same distance can be computed by optimizing over unitary operators:
\begin{equation}
    D(\rho,\sigma)
    =
    \min_{U:U\rho U^\dagger=\sigma}
    \|i\ln U\|.
    \label{eq:charging_distance_unitary}
\end{equation}
The unitary attaining this minimum is denoted by $U_{\rm opt}$ and is called the optimal charging unitary. The corresponding optimal driving Hamiltonian is time independent and given by $V_{\rm opt}=\frac{i\ln U_{\rm opt}}{D(\rho,\sigma)}.$

\section{Geometric power capacity of coherent ergotropy}\label{sect3}
We now apply the quantum charging distance to the coherent part of the work extraction process.
As discussed in the previous section, after extracting the incoherent ergotropy, the state \(\rho\) is transformed into \(\sigma_\rho\) by an incoherent unitary operation.
The remaining extractable work is the coherent ergotropy, $\mathcal{E}_c(\rho) = \Tr\!\left[H_0(\sigma_\rho-P_\rho)\right]$, where \(P_\rho\) is the passive state associated with \(\rho\).

\begin{definition}
\label{definition1}
Let \(\rho\) be a quantum battery state. Let \(\sigma_\rho\) be the active state obtained after the incoherent ergotropy has been extracted, and \(P_\rho\) the passive state associated with \(\rho\), both defined with respect to the same battery Hamiltonian \(H_0\). The coherent ergotropy extraction distance of \(\rho\) is defined as the quantum charging distance between \(\sigma_\rho\) and \(P_\rho\):
\begin{equation}
    D_c^{\rm ext}(\rho)
    =
    D(\sigma_\rho,P_\rho)
    =
    \min_{V_t:\,\|V_t\|=1} T .
    \label{eq:coherent_extraction_distance}
\end{equation}
\end{definition}

\begin{lemma}
\label{lemma1}
For any quantum battery state \(\rho\), the coherent extraction distance $D_c^{\rm ext}(\rho)=D(\sigma_\rho,P_\rho)$
is well defined. Moreover,
\begin{equation*}
    D_c^{\rm ext}(\rho)\geq 0,
    \qquad
    D_c^{\rm ext}(\rho)=0
    \quad\Longleftrightarrow\quad
    \sigma_\rho=P_\rho .
    \label{eq:Dc_zero_condition}
\end{equation*}
\end{lemma}

\noindent{\bf Proof}~~By construction, \(\sigma_\rho\) is obtained from \(\rho\) through an incoherent unitary transformation.
Hence, $\operatorname{spec}(\sigma_\rho)  =\operatorname{spec}(\rho)$.
On the other hand, \(P_\rho\) is the passive state obtained by rearranging the eigenvalues of \(\rho\) in the energy eigenbasis. Therefore, $\operatorname{spec}(P_\rho) =\operatorname{spec}(\rho)$.
It follows that $\operatorname{spec}(\sigma_\rho) = \operatorname{spec}(P_\rho)$,
so \(\sigma_\rho\) and \(P_\rho\) have the same spectrum lying on the same isospectral manifold.
Thus, the quantum charging distance \(D(\sigma_\rho,P_\rho)\) is well defined.

Since \(D\) is a metric on each isospectral manifold, it is non-negative and vanishes if and only if its two arguments coincide. Consequently, $D_c^{\rm ext}(\rho)\geq 0,\quad D_c^{\rm ext}(\rho)=0\Longleftrightarrow
\sigma_\rho=P_\rho$. This proves the claim.
$\Box$.

Therefore, \(D_c^{\rm ext}(\rho)\) is the quantum charging distance restricted to the coherent extraction process
\(\sigma_\rho \to P_\rho\).
It quantifies the minimal geometric distance required to extract the coherent part of the ergotropy.
In this definition, the battery Hamiltonian \(H_0\) and the driving Hamiltonian \(V_t\) play different roles.
The fixed Hamiltonian \(H_0\) determines the energy eigenbasis, the passive state \(P_\rho\), the active state \(\sigma_\rho\), and the amount of extracted coherent work.
By contrast, \(V_t\) is the control Hamiltonian that generates the unitary path from \(\sigma_\rho\) to \(P_\rho\).
The normalization condition \(\|V_t\|=1\) fixes the available driving resource and prevents the extraction time from being trivially shortened by rescaling the driving strength.
Thus, the work is evaluated with respect to \(H_0\), whereas the extraction distance, or equivalently the normalized extraction time, is evaluated with respect to the normalized driving Hamiltonian \(V_t\).

The geometric meaning of \(D_c^{\rm ext}(\rho)\) is illustrated in Fig.~\ref{qubitbloch} for a qubit battery. After the incoherent part has been extracted, the remaining coherent extraction process is the unitary transformation
\(\sigma_\rho\to P_\rho\). The distance \(D_c^{\rm ext}(\rho)\) is the minimal unitary distance associated with this process.
\begin{figure}
\centering
\includegraphics[width=0.8\linewidth]{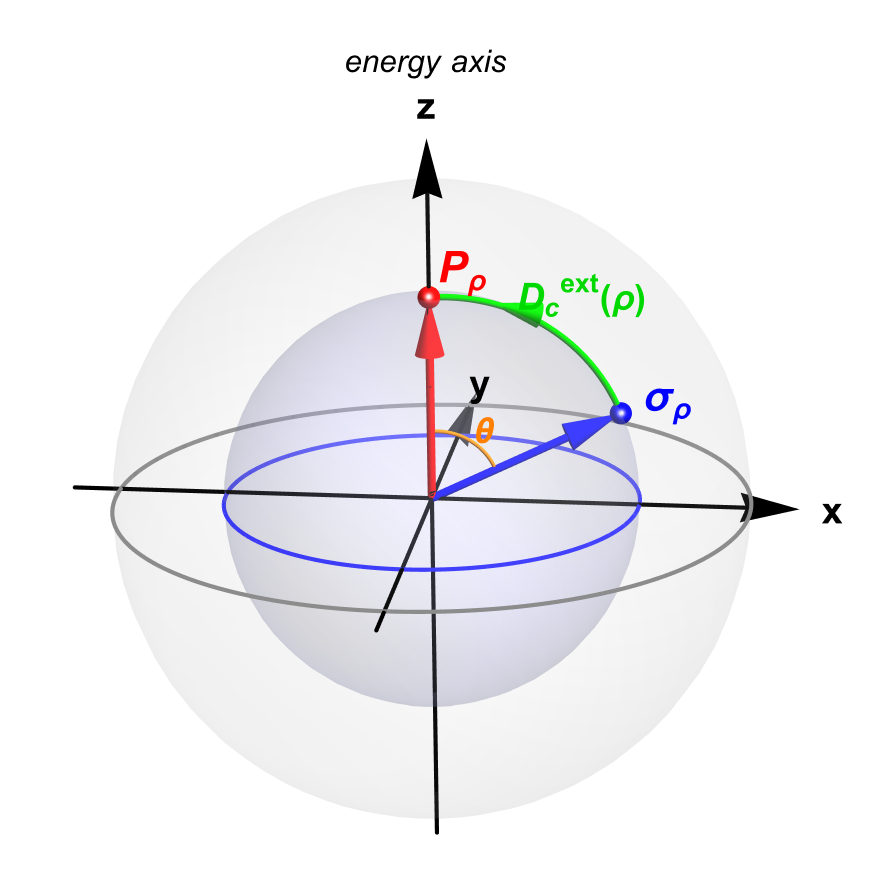}
\caption{
Geometric illustration of the coherent extraction distance \(D_c^{\rm ext}(\rho)\) for a qubit battery.
The green arc shows the minimal unitary distance from the active state \(\sigma_\rho\) to the passive state \(P_\rho\), with the \(z\)-axis defined by \(H_0\).
}
\label{qubitbloch}
\end{figure}

The coherent ergotropy extraction distance has a direct operational meaning in the coherent work extraction process. The coherent ergotropy $\mathcal{E}_c(\rho)$ quantifies the amount of work that can be extracted from the coherence of the initial state, while the coherent ergotropy extraction distance $D_c^{\rm ext}(\rho)$ quantifies the minimal time required to extract this contribution under the normalized driving condition $\|V_t\|=1$. Therefore, the pair $\left(\mathcal{E}_c(\rho),D_c^{\rm ext}(\rho)\right)$ characterizes not only the energetic contribution of coherence, but also the geometric time cost associated with its extraction. This naturally leads to a power-like quantity, obtained by comparing the coherent ergotropy with the shortest time needed to extract it.

\begin{definition}
\label{definition2}
Let $\rho$ be the initial state of a quantum battery with nonzero coherent ergotropy extraction distance, $D_c^{\rm ext}(\rho)>0$. The geometric power capacity of coherent ergotropy is defined as
\begin{equation}
\Pi_c(\rho)=\frac{\mathcal{E}_c(\rho)}{D_c^{\rm ext}(\rho)}=\frac{\Tr\!\left[H_0(\sigma_\rho-P_\rho)\right]}{D(\sigma_\rho,P_\rho)},
\label{eq:geometric_power_capacity}
\end{equation}
where $\mathcal{E}_c(\rho)$ is the coherent ergotropy, and $D_c^{\rm ext}(\rho)=D(\sigma_\rho,P_\rho)$
is the coherent ergotropy extraction distance.
\end{definition}

$\Pi_c(\rho)$ characterizes the maximal average rate at which the coherent part of ergotropy can be extracted when the driving Hamiltonian is normalized by $\|V_t\|=1$.
It combines two aspects of coherence-assisted work extraction: the numerator $\mathcal{E}_c(\rho)$ measures the amount of extractable work originating from coherence, whereas the denominator $D_c^{\rm ext}(\rho)$ measures the minimal geometric time cost required to extract it. Thus, $\Pi_c(\rho)$ provides a geometric power capacity for coherent ergotropy,
see the schematic in Fig.~\ref{schematic}.  
\begin{figure}
\centering
\includegraphics[width=1.0\linewidth]{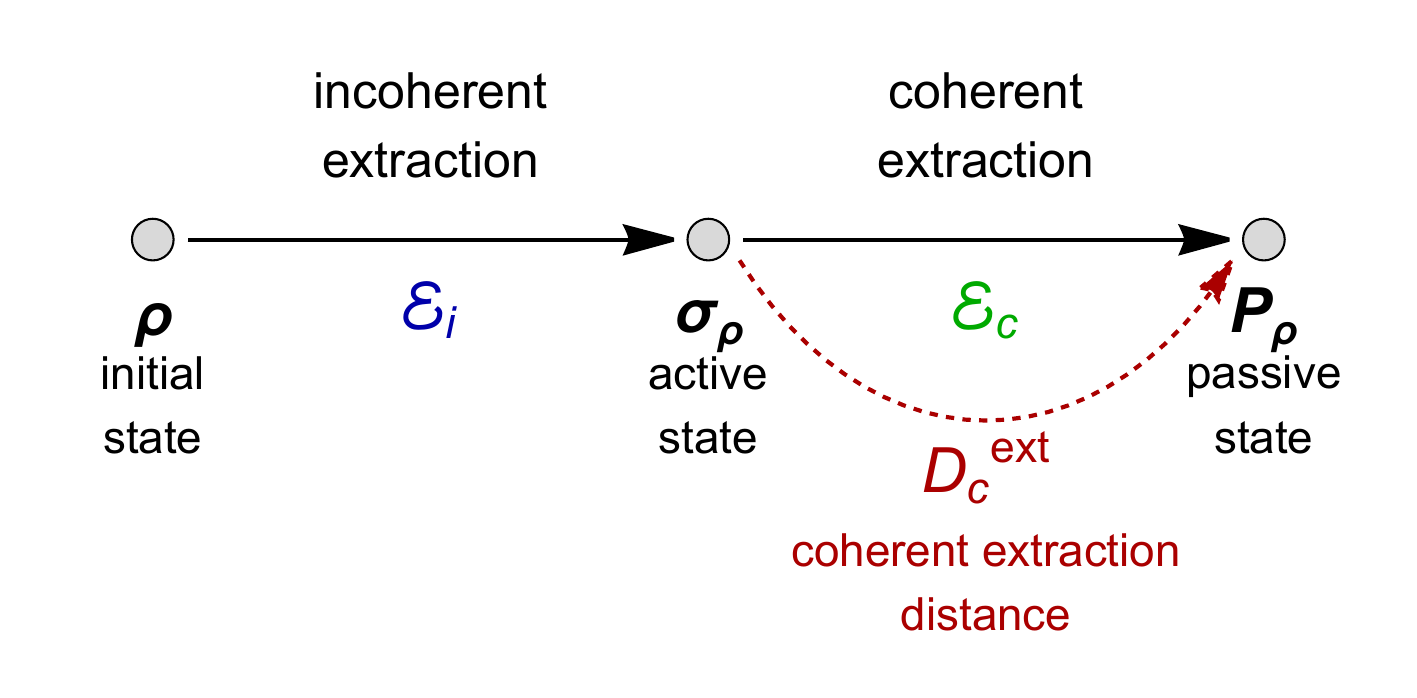}
\caption{The initial state $\rho$ yields active state $\sigma_{\rho}$ after incoherent ergotropy extraction. The unitary process $\sigma_{\rho} \to P_{\rho}$ extracts coherent ergotropy $\mathcal{E}_{c}$ over distance $D_{c}^{\text{ext}}$, defining the geometric power capacity $\Pi_{c}(\rho) = \mathcal{E}_{c} / D_{c}^{\text{ext}}$.}
\label{schematic}
\end{figure}



\begin{definition}
Let $V_t$ be a driving Hamiltonian that realizes the coherent extraction process
$\sigma_\rho \longrightarrow P_\rho=U_T\sigma_\rho U_T^\dagger$,
within a finite time $T_c^{\rm ext}$, where $U_T=\mathcal{T}\exp\!\left(-i\int_0^{T_c^{\rm ext}} V_t\,dt\right)$.
The coherent discharging power of this protocol is defined as
\begin{equation}
    P_c^{\rm ext}(\rho;V_t)
    :=\frac{\mathcal{E}_c(\rho)}{T_c^{\rm ext}},
    \label{eq:actual_coherent_power}
\end{equation}
where
$\mathcal{E}_c(\rho) =\Tr\!\left[H_0(\sigma_\rho-P_\rho)\right]$ is the coherent ergotropy.
\end{definition}

The above definition depends on the concrete driving protocol through the actual extraction time $T_c^{\rm ext}$.
By contrast, the coherent ergotropic power capacity
$\Pi_c(\rho) =\frac{\mathcal{E}_c(\rho)}{D_c^{\rm ext}(\rho)}$
is defined using the coherent ergotropy extraction distance
$D_c^{\rm ext}(\rho)=D(\sigma_\rho,P_\rho)$, which is the minimal extraction time under the unit driving normalization $\|V_t\|=1$.

\begin{lemma}
For any coherent extraction protocol generated by a driving Hamiltonian satisfying the uniform operator-norm constraint,
$\|V_t\|\leq \nu$ with $0\leq t\leq T_c^{\rm ext}$, the actual coherent discharging power is bounded by
\begin{equation}
P_c^{\rm ext}(\rho;V_t)=\frac{\mathcal{E}_c(\rho)}{T_c^{\rm ext}} \leq\nu\,\frac{\mathcal{E}_c(\rho)}{D_c^{\rm ext}(\rho)} =\nu\,\Pi_c(\rho).
 \label{eq:actual_power_bound}
\end{equation}
In particular, under the unit driving convention $\|V_t\|\leq 1$, one obtains
\begin{equation}
    P_c^{\rm ext}(\rho;V_t)
    \leq
    \Pi_c(\rho).
    \label{eq:unit_norm_power_bound}
\end{equation}
\end{lemma}

\noindent {\bf Proof}~~The coherent ergotropy extraction distance $D_c^{\rm ext}(\rho)$ is defined as the minimal time required to implement the transformation $\sigma_\rho\longrightarrow P_\rho$ under the unit driving normalization $\|V_t\|=1$.
If a driving Hamiltonian satisfies $\|V_t\|\leq \nu$, then rescaling the driving strength by $\nu$ implies that the extraction time cannot be smaller than
\begin{equation}
    T_c^{\rm ext}
    \geq
    \frac{D_c^{\rm ext}(\rho)}{\nu}.
    \label{eq:coherent_time_bound}
\end{equation}
Therefore,
\begin{equation}
    P_c^{\rm ext}(\rho;V_t)
    =\frac{\mathcal{E}_c(\rho)}{T_c^{\rm ext}}
    \leq    \nu\,
    \frac{\mathcal{E}_c(\rho)}{D_c^{\rm ext}(\rho)}
    =\nu\,\Pi_c(\rho).
\end{equation}
For $\nu=1$, this reduces to \eqref{eq:unit_norm_power_bound}.
$\Box$

Thus, $\Pi_c(\rho)$ should be understood as the coherent ergotropic power capacity per unit driving strength, rather than the actual power of a particular protocol.
The actual coherent discharging power reaches the bound only when the coherent extraction process is implemented by an optimal driving Hamiltonian realizing the distance $D(\sigma_\rho,P_\rho)$.
We illustrate the above framework by analytically solvable battery systems.

\subsection{Two--level system}\label{subsec:qubit_example}
Consider a two-level quantum battery with Hamiltonian,
$H_0=\epsilon\ket{1}\bra{1}$, where $\epsilon>0$ is the energy gap.
We write the initial state in the energy eigenbasis as
\begin{equation}
    \rho
    =
    \frac{1}{2}
    \left(
    I+x\sigma_x+y\sigma_y+z\sigma_z
    \right),
    \label{eq:qubit_bloch_main}
\end{equation}
where $c=\sqrt{x^2+y^2},\quad R=\sqrt{c^2+z^2}.\label{eq:c_R_main}$
Here, $c$ quantifies the coherence amplitude in the energy basis, while $z$ represents the population imbalance. The physical region is given by $c^2+z^2\leq 1$.

For this qubit battery, the coherent ergotropy, the coherent ergotropy extraction distance, and the geometric power capacity for coherent ergotropy  admit the following closed expressions, see the derivations of in Appendix~\ref{app:qubit_derivation},
\begin{equation}
    \mathcal{E}_c(\rho)
    =\frac{\epsilon}{2}
    \left( R-|z|\right),
    \label{eq:qubit_Ec_main}
\end{equation}
\begin{equation}
    D_c^{\rm ext}(\rho)
    =\frac{1}{2}\arccos\!\left(\frac{|z|}{R}\right),
    \label{eq:qubit_Dc_main}
\end{equation}
and
\begin{equation}
    \Pi_c(\rho)
    =\frac{\mathcal{E}_c(\rho)}{D_c^{\rm ext}(\rho)}
    =\epsilon\frac{R-|z|}{\arccos(|z|/R)} .
    \label{eq:qubit_Pi_main}
\end{equation}
For the incoherent case $c=0$, both $\mathcal{E}_c(\rho)$ and $D_c^{\rm ext}(\rho)$ vanish, and we set $\Pi_c(\rho)=0$ by continuity.

Eq.~\eqref{eq:qubit_Ec_main} shows that coherent ergotropy is generated by the transverse Bloch component $c$.
For fixed population imbalance $z$, increasing $c$ increases the Bloch radius $R$, and hence increases the coherent part of extractable work.
Eq.~\eqref{eq:qubit_Dc_main} gives the geometric extraction cost: it is one half of the angle between the Bloch vectors of $\sigma_\rho$ and $P_\rho$ on the isospectral Bloch sphere.
Thus, Eq.~\eqref{eq:qubit_Pi_main} combines the energetic contribution of coherence with the minimal geometric time cost needed to extract it. If the coherent extraction process is implemented by a driving Hamiltonian satisfying $\|V_t\|\leq \nu$, then the actual coherent discharging power satisfies
$P_c^{\rm ext}(\rho;V_t)\leq\nu\,\Pi_c(\rho).$
In particular, under the unit driving convention $\|V_t\|\leq 1$, one obtains
$P_c^{\rm ext}(\rho;V_t)\leq\Pi_c(\rho)$. Therefore, in this qubit example, $\Pi_c(\rho)$ gives an explicit upper bound on the coherent discharging power with unit driving strength.

Fig.~\ref{fig:qubit_example} summarizes the qubit  battery in the parameter space spanned by the coherence amplitude $c$ and the population imbalance $z$. The left panel shows the coherent ergotropy [Eq.~\eqref{eq:qubit_Ec_main}]. It illustrates that the coherent contribution to extractable work is not determined by coherence alone, but also depends on the population imbalance. In particular, increasing the transverse coherent component generally increases the amount of coherent ergotropy available for extraction.
The middle panel shows the coherent ergotropy extraction distance [Eq.~\eqref{eq:qubit_Dc_main}]. This quantity captures the minimal geometric cost of transforming $\sigma_\rho$ into the passive state $P_\rho$. Unlike the coherent ergotropy, it measures the unitary rotation cost required to extract the coherent part of the work.
The right panel shows the geometric power capacity for coherent ergotropy [Eq.~\eqref{eq:qubit_Pi_main}]. It combines the energetic contribution shown in the left panel with the geometric extraction cost shown in the middle panel. Thus, it identifies the regions where coherence is not only energetically useful, but can also be extracted efficiently under the unit driving constraint.
\begin{figure*}[t]
    \centering
    \includegraphics[width=\linewidth]{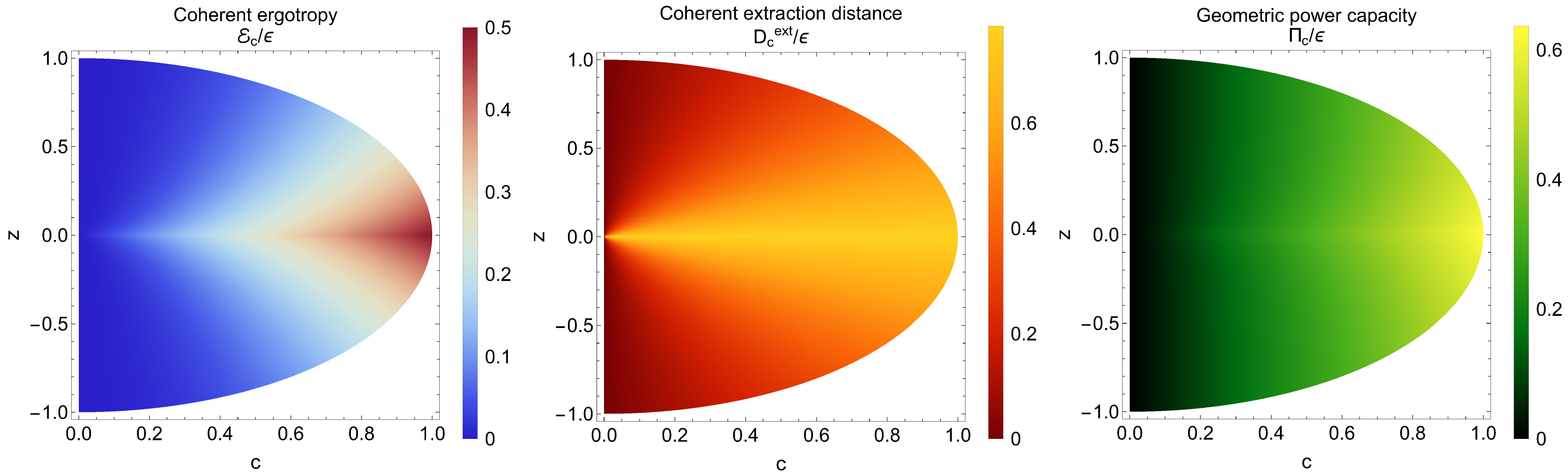}
    \caption{
    Analytical qubit example. In this work, we take $\epsilon=1$.
    The three panels show, respectively, the coherent ergotropy $\mathcal{E}_c(\rho)/\epsilon$[Eq.~\eqref{eq:qubit_Ec_main}], the coherent ergotropy extraction distance $D_c^{\rm ext}(\rho)$[Eq.~\eqref{eq:qubit_Dc_main}], and the geometric power capacity for coherent ergotropy $\Pi_c(\rho)/\epsilon$ [Eq.~\eqref{eq:qubit_Pi_main}] as functions of the coherence amplitude $c=\sqrt{x^2+y^2}$ and the population imbalance $z$.
 }
    \label{fig:qubit_example}
\end{figure*}

\subsection{Three-level system}\label{subsec:qutrit_example}

We next consider a qutrit battery as a simple higher-dimensional example.
Let the battery Hamiltonian be
\begin{equation}
    H_0=\sum_{k=1}^{3}
    \epsilon_k
    \ket{\epsilon_k}\bra{\epsilon_k},~~
    \epsilon_1<\epsilon_2<\epsilon_3.
    \label{eq:qutrit_H0}
\end{equation}
In the energy eigenbasis, we consider the state
\begin{equation}
    \rho    =
    \begin{pmatrix}
        g_1 & c & 0\\
        c^* & g_3 & 0\\
        0 & 0 & g_2
    \end{pmatrix},~~
    g_i=
    \frac{e^{-\beta\epsilon_i}}
    {\sum_j e^{-\beta\epsilon_j}},~~
    |c|\leq \sqrt{g_1g_3}\,.
    \label{eq:qutrit_state}
\end{equation}
For positive temperature, one has $g_1\geq g_2\geq g_3$.
Thus, the diagonal part of $\rho$ is not passively ordered, since the populations appear as $(g_1,g_3,g_2)$.
At the same time, the state contains coherence in the two-dimensional subspace spanned by $\ket{\epsilon_1}$ and $\ket{\epsilon_2}$.

The passive state associated with $\rho$ is
$    P_\rho
    =
    \operatorname{diag}
    \left(
    \lambda_+,g_2,\lambda_-
    \right)$,
where $\lambda_\pm=\frac{g_1+g_3}{2}\pm\Delta$ with $\Delta =\sqrt{\frac{(g_1-g_3)^2}{4}+|c|^2}$.
The incoherent extraction process rearranges the diagonal populations into passive order and transforms $\rho$ into
\begin{equation}
    \sigma_\rho
    =
    \begin{pmatrix}
        g_1 & 0 & c\\
        0 & g_2 & 0\\
        c^* & 0 & g_3
    \end{pmatrix}.
    \label{eq:qutrit_sigma}
\end{equation}
The state $\sigma_\rho$ has the same spectrum as $\rho$ and $P_\rho$, while the remaining coherence is now located in the $\{\ket{\epsilon_1},\ket{\epsilon_3}\}$ subspace.
Therefore, the coherent extraction process $\sigma_\rho\rightarrow P_\rho$
is effectively a two-level rotation in this subspace.

The coherent ergotropy of $\rho$ is
\begin{equation}
    \mathcal{E}_c(\rho)
    =(\epsilon_3-\epsilon_1)\left[\Delta-\frac{g_1-g_3}{2}\right].
    \label{eq:qutrit_Ec}
\end{equation}
The corresponding coherent ergotropy extraction distance is
\begin{equation}
    D_c^{\rm ext}(\rho)
    =\frac{1}{2}
    \arccos\left(\frac{g_1-g_3}{2\Delta}\right).
    \label{eq:qutrit_Dc}
\end{equation}
Hence, the geometric power capacity of coherent ergotropy is
\begin{equation}
\Pi_c(\rho)
=\frac{\mathcal{E}_c(\rho)}{D_c^{\rm ext}(\rho)}
=\frac{(\epsilon_3-\epsilon_1)\left[\Delta-\frac{g_1-g_3}{2}\right]}
{\frac{1}{2}\arccos\left(\frac{g_1-g_3}{2\Delta}\right)}.\label{eq:qutrit_Pic}
\end{equation}
The derivations of Eqs.~\eqref{eq:qutrit_Ec}, \eqref{eq:qutrit_Dc}, and \eqref{eq:qutrit_Pic} are provided in Appendix~\ref{app:qutrit_example}.

It is convenient to introduce the normalized coherence strength
\begin{equation}
\eta=\frac{|c|}{\sqrt{g_1g_3}},\qquad0\leq \eta\leq 1.
\label{eq:qutrit_eta}
\end{equation}
The value $\eta=0$ corresponds to the incoherent case, while $\eta=1$ saturates the positivity constraint of the coherent block.
In terms of $(\beta,\eta)$, we have $|c|=\eta\sqrt{g_1(\beta)g_3(\beta)}$,
and $\Delta(\beta,\eta)=\sqrt{\frac{[g_1(\beta)-g_3(\beta)]^2}{4}+\eta^2 g_1(\beta)g_3(\beta)}$.
Therefore, $\mathcal{E}_c$, $D_c^{\rm ext}$ and $\Pi_c$ can be regarded as functions of the two parameters $\beta$ and $\eta$.
Fig.~\ref{fig:qutrit_example} shows how the coherent ergotropy
\(\mathcal{E}_c\) [Eq.~\eqref{eq:qutrit_Ec}], the coherent extraction distance
\(D_c^{\rm ext}\) [Eq.~\eqref{eq:qutrit_Dc}], and  the geometric power capacity of coherent ergotropy \(\Pi_c=\mathcal{E}_c/D_c^{\rm ext}\) [Eq.~\eqref{eq:qutrit_Pic}] depend on the normalized coherence strength \(\eta\) for fixed inverse temperatures \(\beta\).
As \(\eta\) increases, both the extractable coherent work and the minimal unitary distance required for its extraction generally increase.
The right panel shows their ratio, indicating that \(\Pi_c\) captures the competition between the amount of coherent work and the geometric cost of releasing it.
\begin{figure*}[t]
    \centering
    \includegraphics[width=\linewidth]{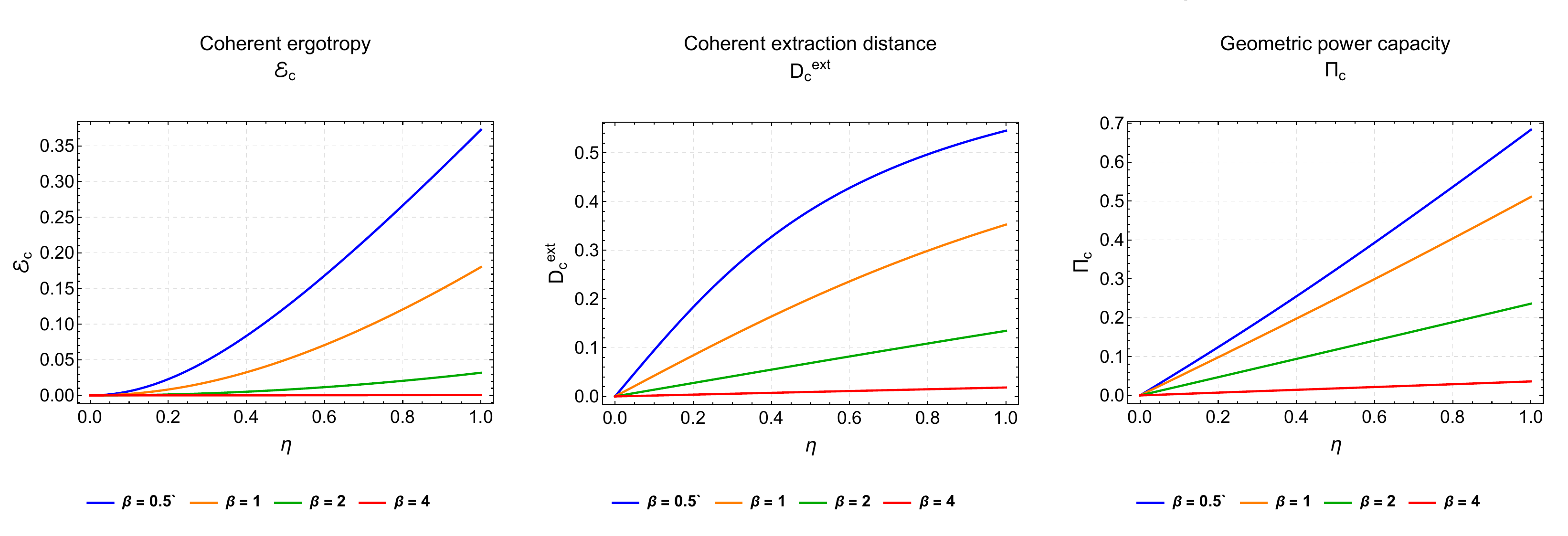}
    \caption{
Qutrit example. The three panels show the coherent ergotropy $\mathcal{E}_c$ [Eq.~\eqref{eq:qutrit_Ec}], the coherent ergotropy extraction distance $D_c^{\rm ext}$ [Eq.~\eqref{eq:qutrit_Dc}], and the geometric power capacity of coherent ergotropy $\Pi_c$ [Eq.~\eqref{eq:qutrit_Pic}] as functions of the normalized coherence strength
$\eta=|c|/\sqrt{g_1g_3}$. Different curves correspond to fixed inverse temperatures $\beta=0.5,1,2,4$.}
\label{fig:qutrit_example}
\end{figure*}

The different curves also show the role of the temperature-dependent population profile.
For smaller $\beta$, the thermal populations are less strongly concentrated in the ground level, so the positivity constraint allows a larger effective coherent block.
As a result, $\mathcal{E}_c$, $D_c^{\rm ext}$, and $\Pi_c$ are generally larger.
For larger $\beta$, the population of the highest level becomes small, which suppresses the allowed coherence and hence reduces all three quantities.

\subsection{Coherent discharging in a three-level battery}
\label{subsec:three_level_coherent_discharge}

We now consider a three-level battery model that admits a simple analytical description of coherent ergotropy extraction.
Consider the battery Hamiltonian,
\begin{equation}
 H_0 =-h\ket{-h}\bra{-h} +h\ket{h}\bra{h},\qquad h>0.
\label{eq:three_level_H0}
\end{equation}
The passive Gibbs state is
\begin{equation}
    \rho_p=\frac{e^{-\beta H_0}}{Z}=
    \operatorname{diag}(a,b,d),
    \label{eq:rho_p_three_level}
\end{equation}
with $a=\frac{e^{\beta h}}{Z}$, $ b=\frac{1}{Z}$, $d=\frac{e^{-\beta h}}{Z}$, $Z=1+2\cosh(\beta h)$. Here the basis is ordered as $\{\ket{-h},\ket{0},\ket{h}\}$, so that $a\geq b\geq d.$

To generate coherence in the energy basis, we introduce the unitary family,
\begin{equation}
    \rho_\theta
    =
    e^{-i\theta X}\rho_p e^{i\theta X},
    \qquad
    0\leq \theta\leq \frac{\pi}{2},
    \label{eq:rho_theta_def}
\end{equation}
where $X =\ket{h}\bra{-h}+\ket{-h}\bra{h}.\label{eq:X_three_level}$
The operator $X$ acts as a Pauli $\sigma_x$ operator in the subspace spanned by $\ket{-h}$ and $\ket{h}$, and satisfies $\|X\|=1$.
The state $\rho_\theta$ takes the explicit form,
\begin{equation}
    \rho_\theta
    =
    \begin{pmatrix}
        a\cos^2\theta+d\sin^2\theta & 0 & i(a-d)\sin\theta\cos\theta\\
        0 & b & 0\\
        -i(a-d)\sin\theta\cos\theta & 0 & a\sin^2\theta+d\cos^2\theta
    \end{pmatrix}.
    \label{eq:rho_theta_matrix}
\end{equation}
Thus, for $0<\theta<\pi/2$, the state contains coherence between the lowest and highest energy levels.

Since $\rho_\theta$ is unitarily connected to the passive state $\rho_p$, the passive state associated with $\rho_\theta$ is still $\rho_p$, namely, $P_{\rho_\theta}=\rho_p$.
Therefore, the extractable work from $\rho_\theta$ is precisely its total ergotropy,
\begin{equation}
    \mathcal{E}(\rho_\theta)
    =
    \Tr\!\left[
    H_0(\rho_\theta-\rho_p)
    \right]
    =
    \frac{4h\sinh(\beta h)}
    {1+2\cosh(\beta h)}
    \sin^2\theta .
    \label{eq:E_theta_three_level}
\end{equation}
This total ergotropy is purely coherent whenever the dephased state $\Delta[\rho_\theta]$ remains passive.
In this case, the incoherent ergotropy vanishes,
$\mathcal{E}_i(\rho_\theta)=0$,
and hence $\mathcal{E}(\rho_\theta)=\mathcal{E}_c(\rho_\theta).$
The coherent-only condition is equivalent to
\begin{equation}
    \sin^2\theta
    \leq
    \frac{1}{e^{\beta h}+1}.
    \label{eq:coherent_only_condition}
\end{equation}
Therefore, we have
\begin{equation}
    \mathcal{E}_c(\rho_\theta)
    =
    \mathcal{E}(\rho_\theta)
    =
    \frac{4h\sinh(\beta h)}
    {1+2\cosh(\beta h)}
    \sin^2\theta
    \label{eq:Ec_theta_three_level}
\end{equation}
for $0\leq \theta\leq \theta_c(\beta)$, where
$$
    \theta_c(\beta)
    =
    \arcsin
    \left[
    \frac{1}{\sqrt{e^{\beta h}+1}}
    \right].
    \label{eq:theta_c}
$$

The coherent ergotropy extraction distance is particularly simple.
Since $\rho_\theta=e^{-i\theta X}\rho_p e^{i\theta X}$ and $\|X\|=1$, the minimal unit-norm driving time from $\rho_\theta$ to $\rho_p$ is
\begin{equation}
    D_c^{\rm ext}(\rho_\theta)
    =
    D(\rho_\theta,\rho_p)
    =
    \theta .
    \label{eq:Dc_theta_three_level}
\end{equation}
Consequently, in the coherent-only region, the geometric power capacity of coherent ergotropy is
\begin{equation}
    \Pi_c(\rho_\theta)
    =
    \frac{\mathcal{E}_c(\rho_\theta)}
    {D_c^{\rm ext}(\rho_\theta)}
    =
    \frac{4h\sinh(\beta h)}
    {1+2\cosh(\beta h)}
    \frac{\sin^2\theta}{\theta}.
    \label{eq:Pi_theta_three_level}
\end{equation} 
At $\theta=0$, we set $\Pi_c(\rho_0)=0$ by continuity.

Fig.~\ref{fig:three_level_coherent_extension} shows the dependence of the coherent extension on the rotation angle $\theta$.
The total ergotropy $\mathcal{E}(\rho_\theta)$ [Eq.~\eqref{eq:Ec_theta_three_level}] increases with $\theta$, whereas the extraction distance is simply $ D(\rho_\theta,\rho_p)=\theta$ [Eq.~\eqref{eq:Dc_theta_three_level}].
Their ratio gives the geometric power capacity of coherent ergotropy [Eq.~\eqref{eq:Pi_theta_three_level}], which measures the extractable work per unit minimal geometric distance.
In the coherent-only region, the dephased state $\Delta[\rho_\theta]$ remains passive, so the incoherent ergotropy vanishes and the total ergotropy is entirely due to energy-basis coherence.
Hence, in this region,
$ \mathcal{E}(\rho_\theta)=\mathcal{E}_c(\rho_\theta),\quad D(\rho_\theta,\rho_p) =D_c^{\rm ext}(\rho_\theta)$.
The right panel therefore gives $\Pi_c(\rho_\theta)$, namely, the geometric power capacity of coherent ergotropy.
\begin{figure*}[t]
    \centering
    \includegraphics[width=\linewidth]{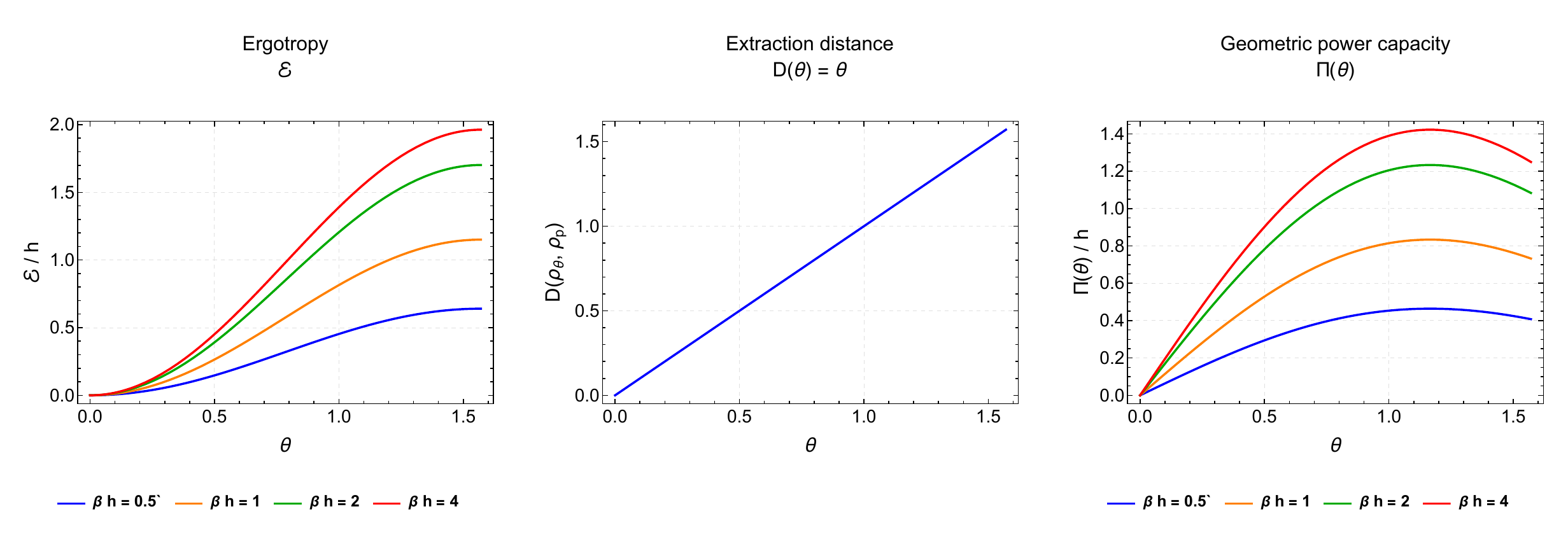}
    \caption{
    Coherent extension of the three-level discharging model.
    The three panels show the total ergotropy $\mathcal{E}(\rho_\theta)$ [Eq.~\eqref{eq:Ec_theta_three_level}], the extraction distance
    $D(\rho_\theta,\rho_p)=\theta$ [Eq.~\eqref{eq:Dc_theta_three_level}], and the corresponding  geometric power capacity of coherent ergotropy $\Pi(\rho_\theta)$ [Eq.~\eqref{eq:Pi_theta_three_level}] as functions of the rotation angle $\theta$. Different curves correspond to different values of $\beta h$. }
    \label{fig:three_level_coherent_extension}
\end{figure*}

\section{Bounds on the geometric power capacity of coherent ergotropy $\Pi_c(\rho)$}
\label{sect4}
In this section, we derive general bounds on the geometric power capacity of coherent ergotropy
$\Pi_c(\rho)=\mathcal{E}_c(\rho)/D_c^{\rm ext}(\rho)$.
The derivation combines two independent ingredients: bounds on the coherent ergotropy and the geometric bounds on the coherent extraction distance.
An upper bound on $\mathcal{E}_c(\rho)$ together with a lower bound on $D_c^{\rm ext}(\rho)$ gives an upper bound on $\Pi_c(\rho)$, while a lower bound on $\mathcal{E}_c(\rho)$ together with an upper bound on $D_c^{\rm ext}(\rho)$ gives a lower bound.

We first recall the  bounds on coherent ergotropy~\cite{coherentergotropyPRL}.
The coherent ergotropy is related to the relative entropy of coherence~\cite{BCP2014},
\begin{equation}
    C_{\rm rel}(\rho)
    =
    S(\Delta[\rho])-S(\rho),
    \label{eq:relative_entropy_coherence_bounds}
\end{equation}
where $S(\rho)=-\Tr(\rho\log\rho)$ is the von Neumann entropy, and $\Delta[\rho]$ denotes the dephasing of $\rho$ in the energy eigenbasis.
Introducing the Gibbs state
\begin{equation}
    \rho_\beta
    =
    \frac{e^{-\beta H_0}}{Z_\beta}
    \label{eq:gibbs_state_bounds}
\end{equation}
with $Z_\beta=\Tr(e^{-\beta H_0})$, one has
\begin{equation}
    \beta\mathcal{E}_c(\rho)
    =
    C_{\rm rel}(\rho)
    +
    D(P_\delta\Vert\rho_\beta)
    -
    D(P_\rho\Vert\rho_\beta).
    \label{eq:Ec_identity_beta}
\end{equation}
Here $D(\cdot\Vert\cdot)$ denotes the quantum relative entropy, $P_\rho$ is the passive state associated with $\rho$, and $P_\delta$ is the passive state associated with the dephased state $\Delta[\rho]$.
Equivalently, $P_\delta$ can be identified with the passive state of the population part of $\rho$.

Since the quantum relative entropy is non-negative, Eq.~\eqref{eq:Ec_identity_beta} implies that
\begin{equation}
    C_{\rm rel}(\rho)-D(P_\rho\Vert\rho_\beta)
    \leq
    \beta\mathcal{E}_c(\rho)
    \leq
    C_{\rm rel}(\rho)+D(P_\delta\Vert\rho_\beta).
    \label{eq:Ec_bounds_raw}
\end{equation}
Equivalently, we have
\begin{equation}
    \mathcal{E}_c^-(\rho;\beta)
    \leq
    \mathcal{E}_c(\rho)
    \leq
    \mathcal{E}_c^+(\rho;\beta),
    \label{eq:Ec_bounds_pm}
\end{equation}
where
\begin{equation}
    \mathcal{E}_c^-(\rho;\beta)
    :=
    \max\left\{
    0,\,
    \frac{
    C_{\rm rel}(\rho)-D(P_\rho\Vert\rho_\beta)
    }{\beta}
    \right\}
    \label{eq:Ec_lower_beta}
\end{equation}
and
\begin{equation}
    \mathcal{E}_c^+(\rho;\beta)
    :=
    \frac{
    C_{\rm rel}(\rho)+D(P_\delta\Vert\rho_\beta)
    }{\beta}.
    \label{eq:Ec_upper_beta}
\end{equation}
The maximum in the lower bound reflects the non-negativity of coherent ergotropy.
These bounds show that the coherent ergotropy is controlled not only by the coherence of the initial state, but also by the thermodynamic distinguishability between the relevant passive states and the Gibbs reference state.

We next bound the $D_c^{\rm ext}(\rho)$. Since $\sigma_\rho$ and $P_\rho$ have the same spectrum, the coherent ergotropy extraction distance $D_c^{\rm ext}(\rho)= D(\sigma_\rho,P_\rho)$
is a special instance of the quantum charging distance between two isospectral states~\cite{Ju-Yeon2024PRA}.
For pure states $\rho=\ket{\psi}\bra{\psi},\quad \sigma=\ket{\phi}\bra{\phi}$,
the quantum charging distance has the following analytic expression,
\begin{equation}
    D(\ket{\psi},\ket{\phi})
    =
    \arccos|\braket{\psi|\phi}|.
    \label{eq:pure_charging_distance_bounds}
\end{equation}
Thus, for pure states, it coincides with the Bures angle, and its maximal value is $\pi/2$, attained by orthogonal states.

For mixed states, the Bures angle is defined as
\begin{equation}
    \theta_B(\rho,\sigma)
    =\arccos \mathcal{F}(\rho,\sigma),
    \label{eq:Bures_angle_def_bounds}
\end{equation}
where $\mathcal{F}(\rho,\sigma)=\Tr\sqrt{\sqrt{\rho}\sigma\sqrt{\rho}}$
is the Uhlmann fidelity. The quantum charging distance satisfies
\begin{equation}
    \theta_B(\rho,\sigma)
    \leq
    D(\rho,\sigma).
    \label{eq:Bures_lower_bound}
\end{equation}
Moreover, for any two $d$-dimensional isospectral states,
\begin{equation}
    D(\rho,\sigma)
    \leq
    \pi\left(1-\frac{1}{d}\right).
    \label{eq:charging_distance_upper_bound}
\end{equation}
Applying these bounds to the pair $(\sigma_\rho,P_\rho)$, we have
\begin{equation}
    D_-(\rho)
    \leq
    D_c^{\rm ext}(\rho)
    \leq
    D_+,
    \label{eq:Dc_bounds_pm}
\end{equation}
where
$D_-(\rho)=\theta_B(\sigma_\rho,P_\rho)$
and $ D_+:=\pi\left(1-\frac{1}{d}\right)$, with $d$ the Hilbert-space dimension of the battery.

From the energetic bounds in \eqref{eq:Ec_bounds_pm} and the geometric bounds in \eqref{eq:Dc_bounds_pm}, we obtain the following bound on the geometric power capacity of coherent ergotropy.
\begin{proposition}
\label{prop:Pi_c_bounds}
For a nontrivial coherent extraction process with $D_c^{\rm ext}(\rho)>0$, the geometric power capacity of coherent ergotropy satisfies
\begin{equation}
    \frac{
    \mathcal{E}_c^-(\rho;\beta)
    }{
    D_+
    }
    \leq
    \Pi_c(\rho)
    \leq
    \frac{
    \mathcal{E}_c^+(\rho;\beta)
    }{
    D_-(\rho)
    }
    \label{eq:Pi_c_bounds_abstract}
\end{equation}
for every finite positive $\beta$ for which $D_-(\rho)>0$.
Equivalently,
\begin{equation}
    \frac{
    \max\left\{
    0,\,
    C_{\rm rel}(\rho)-D(P_\rho\|\rho_\beta)
    \right\}
    }{
    \beta\,\pi\left(1-\frac{1}{d}\right)
    }
    \leq
    \Pi_c(\rho)
    \leq
    \frac{
   C_{\rm rel}(\rho)+D(P_\delta\|\rho_\beta)
    }{
    \beta\,\theta_B(\sigma_\rho,P_\rho)
    }.
    \label{eq:Pi_c_bounds_explicit}
\end{equation}
\end{proposition}

The bound above has a simple interpretation.
The numerator is controlled by thermodynamic and coherence-theoretic quantities, namely, the relative entropy of coherence and the relative-entropy distances from the passive states to the Gibbs reference state.
The denominator is controlled by the geometry of the isospectral manifold on which the coherent extraction process takes place.
Therefore, $\Pi_c(\rho)$ is constrained simultaneously by the amount of coherent work available and by the minimal geometric distance required to extract it.

One can further optimize over the auxiliary inverse temperature $\beta$.
Define
\begin{equation*}
\mathcal{E}_{c,\mathrm{opt}}^-:= \sup_{\beta>0}\mathcal{E}_c^-(\rho;\beta),~~
\mathcal{E}_{c,\mathrm{opt}}^+:=\inf_{\beta>0}\mathcal{E}_c^+(\rho;\beta).
\end{equation*}
The corresponding optimized bounds on the geometric power capacity of coherent ergotropy are
\begin{equation}
\frac{ \mathcal{E}_{c,\mathrm{opt}}^-}{\pi\left(1-\frac{1}{d}\right)}
\leq\Pi_c(\rho)\leq\frac{ \mathcal{E}_{c,\mathrm{opt}}^+}{\theta_B(\sigma_\rho,P_\rho)}.
 \label{eq:Pi_c_optimized_bounds}
\end{equation}
This form makes clear that the geometric power capacity of coherent ergotropy is neither determined solely by the coherent ergotropy nor solely by the extraction distance.
It is a resource--geometry ratio: the coherent ergotropy quantifies the available coherent work, whereas the coherent extraction distance quantifies the minimal geometric cost required to release it.

Finally, using the relation between the geometric power capacity of coherent ergotropy and the actual coherent discharging power, one has that if the driving Hamiltonian satisfies $\|V_t\|\leq \nu$,
then $P_c^{\rm ext}(\rho;V_t)\leq \nu\,\Pi_c(\rho)$. Together with the upper bound in Eq.~\eqref{eq:Pi_c_optimized_bounds}, this gives rise to
\begin{equation}
P_c^{\rm ext}(\rho;V_t) \leq \nu\,
 \frac{\mathcal{E}_{c,\mathrm{opt}}^+}{\theta_B(\sigma_\rho,P_\rho)}.
\label{eq:actual_power_upper_bound}
\end{equation}
Thus, the actual coherent discharging power is bounded by both the available coherent work and the geometric cost of implementing the coherent extraction process.

Detailed illustrations of these bounds for the qubit and qutrit
families introduced in  Secs.~\ref{subsec:qubit_example} and ~\ref{subsec:qutrit_example} are provided in
Appendix~\ref{app:general_bound_examples}.
The qubit example shows that the energetic upper bound can become exact
for an appropriate Gibbs reference, while the qutrit example illustrates
how the energetic and geometric constraints jointly determine the bounds
on \(\Pi_c(\rho)\).

\section{Bounds on the geometric power capacity of coherent ergotropy based on different coherence measures}\label{sect5}
The bounds derived above are expressed in terms of the relative entropy of coherence
measure $C_{\rm rel}(\rho)$. In fact, other coherence
measures may give rise to clearer operational meanings or tighter bounds in different regimes.
Generally, the same construction can be applied to any coherence measure $M(\rho)$ that bounds the relative entropy of coherence as
\begin{equation}
    L_M(\rho)
    \leq
    C_r(\rho)
    \leq
    U_M(\rho),
    \label{eq:generic_coherence_measure_bounds}
\end{equation}
such that when $M=C_{\rm rel}$, $L_M(\rho)=U_M(\rho)=C_{\rm rel}(\rho)$.
Then Eq.~\eqref{eq:Ec_identity_beta} implies the $M$-induced coherent ergotropy bounds
\begin{equation}
    \mathcal{E}_{c,M}^-(\rho;\beta)
    :=
    \max\left\{
    0,\,
    \frac{
    L_M(\rho)-D_{\rm rel}(P_\rho\|\rho_\beta)
    }{\beta}
    \right\}
    \label{eq:Ec_M_lower}
\end{equation}
and
\begin{equation}
    \mathcal{E}_{c,M}^+(\rho;\beta)
    :=
    \frac{
    U_M(\rho)+D_{\rm rel}(P_\delta\|\rho_\beta)
    }{\beta}.
    \label{eq:Ec_M_upper}
\end{equation}
Consequently,
\begin{equation}
    \frac{
    \mathcal{E}_{c,M}^-(\rho;\beta)
    }{
    D_+
    }
    \leq
    \Pi_c(\rho)
    \leq
    \frac{
    \mathcal{E}_{c,M}^+(\rho;\beta)
    }{
    D_-(\rho)
    }.
    \label{eq:Pi_c_M_induced_bounds}
\end{equation}
\eqref{eq:Pi_c_M_induced_bounds} shows that different coherence measures induce different  bounds on the geometric power capacity of coherent ergotropy.
Therefore, the choice of coherence measure affects not only the estimation of coherent ergotropy itself, but also the estimation of how efficiently this coherent work can be geometrically extracted.

Consider the experimentally accessible $l_1$-coherence measure \cite{BCP2014},
\begin{equation}
    C_{l_1}(\rho)=\sum_{i\neq j}|\rho_{ij}|.
\end{equation}
From the bounds \cite{Rana2017PRA}
\begin{equation}
    \frac{C_{l_1}^2(\rho)}{2d(d-1)}
    \leq
    C_r(\rho)
    \leq
    \log[1+C_{l_1}(\rho)] ,
\end{equation}
we identify in Eq.~\eqref{eq:generic_coherence_measure_bounds},
\begin{equation}
    L_{l_1}(\rho)
    =
    \frac{C_{l_1}^2(\rho)}{2d(d-1)},
 ~~   U_{l_1}(\rho)
    =
    \log[1+C_{l_1}(\rho)] .
\end{equation}
Substituting them into Eqs.~\eqref{eq:Ec_M_lower}--\eqref{eq:Pi_c_M_induced_bounds}, we obtain the \(l_1\)-coherence-induced bounds on \(\Pi_c(\rho)\).

For the robustness of coherence \cite{Napoli2016PRL,Piani2016PRA},
\begin{align}
C_R(\rho) &:= \min_{\sigma} \left\{ s \geqslant 0 \,\middle|\, \frac{\rho + s \sigma}{1 + s} \in\mathcal{I}\right\} \nonumber\\
&= \min_{\tau \in \mathcal{I}} \{ s \geqslant 0 \,|\, \rho \leqslant (1 + s)\tau \},\nonumber
\end{align}
where $\mathcal{I}$ denotes the set of incoherent states in the energy eigenbasis, one has~\cite{Rana2017PRA}
\begin{equation}
    C_{\rm rel}(\rho)
    \leq
    \log\!\left[1+C_R(\rho)\right]
    \leq
    \log\!\left[1+C_{\ell_1}(\rho)\right].
    \label{eq:Cr_robustness_bound}
\end{equation}
Therefore, by taking $U_R(\rho)=\log\!\left[1+C_R(\rho)\right]$
we obtain the robustness-induced upper bound,
\begin{equation}
    \Pi_c(\rho)
    \leq
    \frac{
    \log\!\left[1+C_R(\rho)\right]
    +
    D_{\rm rel}(P_\delta\|\rho_\beta)
    }{
    \beta D_-(\rho)
    }.
    \label{eq:Pi_robustness_upper}
\end{equation}
Since \(C_R(\rho)\leq C_{\ell_1}(\rho)\), this upper bound is no looser than
the corresponding \(l_1\)-coherence-induced upper bound.

Next, consider a fidelity-based geometric coherence measure. Let
   $ \mathcal F_{\mathcal I}(\rho)
    :=
    \max_{\delta\in\mathcal I}
    \mathcal F(\rho,\delta),
    \label{eq:FI_definition}$
where $\mathcal F(\rho,\delta)=\left\|\sqrt{\rho}\sqrt{\delta}\right\|_1=\Tr\sqrt{\sqrt{\rho}\sigma\sqrt{\rho}}$
is the Uhlmann fidelity. The geometric coherence can be written as~\cite{Streltsov2017RMP,Streltsov2015PRL},
\begin{equation}
    C_g(\rho)
    =
    1-\mathcal F_{\mathcal I}^2(\rho).
    \label{eq:geometric_coherence}
\end{equation}

The quantum relative entropy satisfies the standard fidelity bound~\cite{MullerLennert2013Renyi},
\begin{equation}
    D_{\rm rel}(\rho\|\delta)
    \geq
    -2\log \mathcal F(\rho,\delta),
    \label{eq:relative_entropy_fidelity_bound}
\end{equation}
following from the monotonicity of the quantum R\'{e}nyi relative entropy with respect to its R\'{e}nyi parameter.
Since the relative entropy of coherence can be written as $C_{\rm rel}(\rho)=\min_{\delta\in\mathcal I} D_{\rm rel}(\rho\|\delta)$,
we obtain
\begin{equation}
    C_{\rm rel}(\rho)
    \geq
    -2\log \mathcal F_{\mathcal I}(\rho).
    \label{eq:Cr_fidelity_lower}
\end{equation}
From Eq.~\eqref{eq:geometric_coherence} this bound can also be written as
\begin{equation}
    C_{\rm rel}(\rho)
    \geq
    -\log\!\left[1-C_g(\rho)\right].
    \label{eq:Cr_Cg_lower}
\end{equation}
Substituting this lower bound ~\eqref{eq:Cr_Cg_lower} into Eq.~\eqref{eq:Ec_M_lower}, we obtain the fidelity-induced lower bound on coherent ergotropy,
\begin{equation}
\begin{aligned}
\mathcal{E}_{c,g}^-(\rho;\beta)
&:=\max\left\{0,\,\frac{-2\log \mathcal{F}_{\mathcal{I}}(\rho) - D_{\mathrm{rel}}(P_\rho\|\rho_\beta)}{\beta}\right\} \\
&=\max\left\{ 0,\,\frac{-\log[1-C_g(\rho)] - D_{\mathrm{rel}}(P_\rho\|\rho_\beta)}{\beta}\right\}.\nonumber
\end{aligned}
\label{eq:Ec_geometric_lower}
\end{equation}
Therefore, the geometric coherence induces the lower bound
\begin{equation*}
\Pi_c(\rho)\geq\frac{\mathcal{E}_{c,g}^-(\rho;\beta) }{ D_+}.
\label{eq:Pi_geometric_lower}
\end{equation*}
This bound has a direct geometric interpretation.
It relates the geometric power capacity of coherent ergotropy to the fidelity distance from the set of energy-incoherent states. In contrast to the relative-entropy-based bound, this construction emphasizes how far the state is from being incoherent in a metric sense.

Different choices of coherence measures $M$ lead to different trade-offs.
The relative entropy of coherence gives the most direct thermodynamic expression.
The $l_1$-coherence is simply evaluated and experimentally accessible.
The robustness of coherence may yield a tighter upper bound, while fidelity-based geometric coherence provides a natural geometric lower estimate. Therefore, the bounds on $\Pi_c(\rho)$ are not only determined by the amount of coherent ergotropy, but also by how the chosen coherence measure constrains the thermodynamic resource and how the quantum charging distance constrains the minimal extraction geometry. 
Since the actual coherent discharging power satisfies
\(P_c^{\rm ext}(\rho;V_t)\leq\nu\Pi_c(\rho)\) under the uniform
constraint \(\|V_t\|\leq\nu\), the upper bounds above also provide
resource-dependent upper bounds on the actual coherent discharging power.
Thus, this hierarchy connects coherence quantification, ergotropic work extraction, and geometric speed limits in a unified way.

Detailed comparisons of these coherence-measure-induced bounds for
qubit and qutrit systems are given in
Appendix~\ref{Comparison of coherence-measure-induced bounds}.
For the qubit example, the relative-entropy upper bound becomes exact
when the Gibbs reference coincides with the passive state, whereas the
\(l_1\)-coherence bounds are simpler but generally less tight.
For the single-block qutrit family at fixed \(\beta\), by contrast,
the \(l_1\)-coherence determines the coherent amplitude and yields the
exact coherent ergotropy.
These examples show that the usefulness of a coherence-measure-induced
bound depends on both the chosen measure and the structure of the state.

\section{Protocol-corrected geometric power capacity} 
\label{sec:protocol_corrected_capacity} 

The geometric power capacity of coherent ergotropy \(\Pi_c(\rho)\), defined in Eq.~\eqref{eq:geometric_power_capacity}, is an intrinsic capacity associated with the coherent extraction process \(\sigma_\rho\to P_\rho\). 
We now establish the relations between this intrinsic quantity and the physical power generated by concrete driving protocols~\cite{Campaioli2018PRL}. 
 
Consider a driving Hamiltonian \(V_t\) that realizes the coherent extraction 
process $\sigma_\rho \longrightarrow P_\rho$ 
within a finite time $T$. The corresponding evolution is 
\begin{equation*} 
    \rho_t 
    = 
    U_t\sigma_\rho U_t^\dagger, 
~~ 
    \dot U_t=-iV_tU_t, 
~~ 
    \rho_0=\sigma_\rho,~~ \rho_T=P_\rho . 
\end{equation*} 
The actual coherent discharging power is then given by 
Eq.~\eqref{eq:actual_coherent_power}. Since the coherent extraction distance 
\(D_c^{\rm ext}(\rho)\) is the minimal unit-norm time required to connect 
\(\sigma_\rho\) and \(P_\rho\), a general driving Hamiltonian naturally 
introduces the effective driving speed~\cite{Ju-Yeon2024PRA,PhysRevX6021031,MT,ML}, 
\begin{equation} 
    v_{\rm eff}[V_t] 
    := 
    \frac{1}{T} 
    \int_0^T \|V_t\|\,dt. 
    \label{eq:v_eff} 
\end{equation} 
The corresponding quantum speed limit implies that 
\begin{equation} 
    T 
    \geq 
    \frac{D_c^{\rm ext}(\rho)} 
    {v_{\rm eff}[V_t]} . 
    \label{eq:QSL_eff_protocol} 
\end{equation} 

Combining this with Eq.~\eqref{eq:actual_coherent_power}, one obtains 
\begin{equation} 
    P_c^{\rm ext}(\rho;V_t) 
    \leq 
    v_{\rm eff}[V_t]\, 
    \frac{\mathcal{E}_c(\rho)} 
    {D_c^{\rm ext}(\rho)} 
    = 
    v_{\rm eff}[V_t]\, 
    \Pi_c(\rho). 
    \label{eq:power_bound_eff} 
\end{equation} 

This gives a simple operator-norm-based protocol-corrected bound. 
For convenience, we denote it by
\begin{equation} 
    \Gamma_c^{\rm eff}[\rho,V_t] 
    := 
    v_{\rm eff}[V_t]\, 
    \Pi_c(\rho), 
    \label{eq:Gamma_eff_def} 
\end{equation} 
which bounds the coherent discharging power generated by the protocol
\(V_t\). It requires only the operator norm of the applied driving Hamiltonian
and can therefore be evaluated without knowledge of the instantaneous
state trajectory.

The above bound can be refined when state-resolved dynamical information
is available.
A Hermitian operator \(K_t\) satisfying 
$ [K_t,\rho_t]=0$ does not contribute to the instantaneous motion of the state, since it leaves \(-i[V_t,\rho_t]\) unchanged up to this redundant component. We therefore define the transverse norm of the driving Hamiltonian by 
\begin{equation} 
    \|V_t\|_{\rho_t,\perp} 
    := 
    \min_{K_t=K_t^\dagger,\,[K_t,\rho_t]=0} 
    \|V_t-K_t\| . 
    \label{eq:transverse_norm} 
\end{equation} 
This gives the reduced effective driving speed, 
\begin{equation} 
    v_{\rm eff}^{\perp}[V_t] 
    := 
    \frac{1}{T} 
    \int_0^T 
    \|V_t\|_{\rho_t,\perp}\,dt . 
    \label{eq:v_eff_perp} 
\end{equation} 
By construction, 
\begin{equation} 
    v_{\rm eff}^{\perp}[V_t] 
    \leq 
    v_{\rm eff}[V_t] . 
\end{equation} 
 
Using the same speed-limit bound with the transverse norm, one obtains 
\begin{equation} 
    P_c^{\rm ext}(\rho;V_t) 
    \leq 
    v_{\rm eff}^{\perp}[V_t]\, 
    \Pi_c(\rho). 
    \label{eq:power_bound_eff_perp} 
\end{equation}

We therefore introduce the corresponding refined protocol-corrected bound,
\begin{equation} 
    \Gamma_c^{\perp}[\rho,V_t] 
    := 
    v_{\rm eff}^{\perp}[V_t]\, 
    \Pi_c(\rho). 
    \label{eq:Gamma_perp_def} 
\end{equation} 
The two protocol-corrected bounds satisfy 
\begin{equation} 
    P_c^{\rm ext}(\rho;V_t) 
    \leq 
    \Gamma_c^{\perp}[\rho,V_t] 
    \leq 
    \Gamma_c^{\rm eff}[\rho,V_t]. 
    \label{eq:Gamma_hierarchy} 
\end{equation} 
The protocol-corrected capacities \(\Gamma_c^{\rm eff}\) and \(\Gamma_c^{\perp}\) incorporate the effective
speed of a specified driving Hamiltonian. They therefore convert the
intrinsic resource-geometric capacity into an upper bound on the
physically achievable coherent extraction power under a given protocol.
Among these two protocol-corrected quantities, \(\Gamma_c^{\perp}\)
provides the sharper bound. 
However, the improvement is not necessarily strict. The two bounds
coincide whenever the transverse minimization does not reduce the
effective driving norm, i.e., $v_{\rm eff}^{\perp}[V_t]=v_{\rm eff}[V_t]$.
The reason is that \(v_{\rm eff}^{\perp}\)
removes the components of the driving Hamiltonian that commute with the
instantaneous state and hence do not generate motion along the coherent
extraction path. These components may contribute to the operator norm of
the Hamiltonian, but not to the actual evolution of the state. Thus,
\(\Gamma_c^{\perp}\) captures the effective driving resource more
faithfully than \(\Gamma_c^{\rm eff}\). 
Nevertheless, the simpler bound \(\Gamma_c^{\rm eff}\) remains useful
when only the control Hamiltonian is known. It requires only
\(\|V_t\|\), whereas evaluating \(\Gamma_c^{\perp}\) additionally
requires the instantaneous state trajectory \(\rho_t\) and the
minimization in Eq.~\eqref{eq:transverse_norm} at each time.
Thus, \(\Gamma_c^{\rm eff}\) is preferable when state-resolved
information is unavailable or costly to obtain.

The intrinsic capacity \(\Pi_c(\rho)\) is recovered in the ideal 
unit-speed optimal-driving limit. If the coherent extraction process is 
implemented by an optimal Hamiltonian with no redundant commuting 
component and with normalized operator norm, then 
$ v_{\rm eff}^{\perp}[V_t] = v_{\rm eff}[V_t] = 1$ . 
Consequently, 
$ \Gamma_c^{\perp}[\rho,V_t] =\Gamma_c^{\rm eff}[\rho,V_t] =\Pi_c(\rho).$ 
Thus, both protocol-corrected bounds reduce to the intrinsic geometric
power capacity in the unit-speed optimal-driving limit.

The above construction can also be combined with the resource-geometric 
bounds derived in Sec.~\ref{sect4}. By using the coherent ergotropy bounds 
in Eq.~\eqref{eq:Ec_bounds_pm} and the coherent extraction distance 
bounds in Eq.~\eqref{eq:Dc_bounds_pm}, the intrinsic geometric power 
capacity obeys Eq.~\eqref{eq:Pi_c_bounds_abstract}. Multiplying by the 
reduced effective driving speed \(v_{\rm eff}^{\perp}[V_t]\), we obtain 
\begin{equation} 
    v_{\rm eff}^{\perp}[V_t]\, 
    \frac{\mathcal{E}_c^{-}(\rho;\beta)} 
    {D_+} 
    \leq 
    \Gamma_c^{\perp}[\rho,V_t] 
    \leq 
    v_{\rm eff}^{\perp}[V_t]\, 
    \frac{\mathcal{E}_c^{+}(\rho;\beta)} 
    {D_-(\rho)} . 
    \label{eq:Gamma_perp_resource_geometric_bounds} 
\end{equation} 
In particular, the actual coherent extraction power satisfies 
\begin{equation} 
    P_c^{\rm ext}(\rho;V_t) 
    \leq 
    v_{\rm eff}^{\perp}[V_t]\, 
    \frac{\mathcal{E}_c^{+}(\rho;\beta)} 
    {D_-(\rho)} . 
    \label{eq:actual_power_resource_geometric_bound} 
\end{equation} 
For example, using the relative entropy coherence $C_r(\rho)$ upper bound~\eqref{eq:Ec_upper_beta}, 
one obtains 
\begin{equation} 
    P_c^{\rm ext}(\rho;V_t) 
    \leq 
    v_{\rm eff}^{\perp}[V_t]\, 
    \frac{ 
    C_r(\rho) 
    + 
    D_{\rm rel}(P_\delta\|\rho_\beta) 
    } 
    {\beta\,D_-(\rho)} . 
    \label{eq:power_bound_relative_entropy_speed} 
\end{equation} 
Moreover, the robustness of coherence $C_R(\rho)$ provides an upper bound on $C_r(\rho)$, as shown in Eq.~\eqref{eq:Cr_robustness_bound}, 
\begin{equation} 
    P_c^{\rm ext}(\rho;V_t) 
    \leq 
    v_{\rm eff}^{\perp}[V_t]\, 
    \frac{ 
    \ln[1+C_R(\rho)] 
    + 
    D_{\rm rel}(P_\delta\|\rho_\beta) 
    } 
    {\beta\,D_-(\rho)} . 
    \label{eq:power_bound_robustness_speed} 
\end{equation} 
 
These inequalities show that the coherent discharging power is jointly 
constrained by three ingredients: the amount of coherence resource, the 
state-space distance required to extract the coherent ergotropy, and the 
effective speed of the driving protocol. The reduced speed 
\(v_{\rm eff}^{\perp}[V_t]\) isolates the part of the driving Hamiltonian 
that actually moves the state along the coherent extraction path, while 
\(\Pi_c(\rho)\) captures the intrinsic geometric efficiency of converting 
coherent ergotropy into extractable work. Therefore, the protocol-corrected 
bound 
$P_c^{\rm ext}(\rho;V_t) 
\leq 
v_{\rm eff}^{\perp}[V_t]\Pi_c(\rho)$ 
provides an operational connection among coherent ergotropy, 
charging-distance geometry, and physically available driving resources.

\section{Conclusion}
\label{sec:conclusion}

We have investigated coherent ergotropy extraction in quantum batteries from a resource-geometric perspective.
We introduced the geometric power capacity of coherent ergotropy, $\Pi_c(\rho)$, which measures the amount of coherent ergotropy extractable unit minimal coherent extraction distance. This quantity combines the coherent work resource $\mathcal{E}_c(\rho)$ with the geometric cost $D_c^{\rm ext}(\rho)$ required to extract it.

We have shown that $\Pi_c(\rho)$ provides a natural upper bound on the actual coherent discharging power under the unit driving convention. More generally, when the driving Hamiltonian has a bounded operator norm, the actual coherent extraction power is bounded by the product of the driving strength and $\Pi_c(\rho)$. Thus, $\Pi_c(\rho)$ should be understood as a capacity under unit driving norm rather than the power of a particular protocol. Once the driving strength or the effective speed of a concrete protocol is specified, this intrinsic quantity acquires an operational meaning.

We have derived general bounds on $\Pi_c(\rho)$ by combining the bounds on coherent ergotropy with the geometric bounds on the coherent extraction distance. On the resource side, coherent ergotropy can be bounded by the relative entropy of coherence and by other coherence measures. On the geometric side, the coherent extraction distance is a special instance of the quantum charging distance and therefore satisfies general distance bounds. This provides a practical way to estimate the geometric power capacity of coherent ergotropy even when the exact coherent ergotropy or the exact extraction distance is difficult to compute.

We have illustrated the framework with analytically tractable qubit and qutrit examples.
For qubit case, the coherent extraction process admits a Bloch-sphere interpretation, where the coherent ergotropy is related to the energetic displacement from the passive direction, and the extraction distance is determined by the angular separation between the active and passive states. For qutrits with a single coherent block, the coherent extraction process effectively reduces to a two-level rotation in a higher-dimensional system. These examples show that coherent ergotropy and coherent extraction distance capture distinct aspects of the discharging process, and that their ratio $\Pi_c(\rho)$ reflects a genuine resource-geometric feature.

We have also formulated the coherence measure induced bounds and protocol-corrected geometric power capacities.
The former connects different coherence monotones, such as relative entropy coherence, $l_1$-coherence, robustness of coherence and fidelity-based geometric coherence, with operational estimates of coherent discharging performance.
The latter relates the capacity $\Pi_c(\rho)$ to concrete driving protocols by incorporating the effective speed of the applied Hamiltonian and removing dynamically irrelevant Hamiltonian components.

Overall, our results show that coherent ergotropic power is jointly constrained by the amount of coherent work stored in the state and the minimal geometric distance required to release it. The geometric power capacity of coherent ergotropy $\Pi_c(\rho)$ combines these two ingredients into a single figure of merit, providing a useful tool for analyzing, bounding, and comparing coherent discharging processes in quantum batteries. Future work may extend this approach to open system dynamics, many-body batteries and experimentally constrained control protocols, where coherence, geometry and driving speed are expected to be strongly intertwined.

\section*{ACKNOWLEDGEMENTS}
This work is supported by the National Natural Science Foundation of China (NSFC) under Grant No. 12171044, the Natural Science Foundation of Jilin Province, China (Grant No. 20250102018JC), and the Specific Research Fund of the Innovation Platform for Academicians of Hainan Province.

\appendix

\subsection{Derivation of the qubit example}\label{app:qubit_derivation}
In this appendix, we derive the closed expressions used in Sec.~\ref{subsec:qubit_example}.
Consider the qubit Hamiltonian $ H_0=\epsilon\ket{1}\bra{1}$.
We use the Bloch representation $\rho=\frac{1}{2}\left(I+x\sigma_x+y\sigma_y+z\sigma_z\right).$
Define $c=\sqrt{x^2+y^2}$ and $R=\sqrt{c^2+z^2}$. The eigenvalues of $\rho$ are $ r_{\pm} =\frac{1\pm R}{2}.$

Since the passive state is obtained by placing the larger eigenvalue on the lower energy level, so $P_\rho =\frac{1}{2}\left(I+R\sigma_z\right)$. Therefore, the Bloch vector of $P_\rho$ is $\mathbf{r}_{P_\rho}=(0,0,R)$.

The state $\sigma_\rho$ is obtained from $\rho$ by the optimal incoherent unitary operation.
This operation preserves the coherence amplitude and rearranges the populations so that the $z$ component becomes non-negative.
Up to an irrelevant phase rotation around the energy axis, we may write
$\sigma_\rho = \frac{1}{2} \left( I+c\sigma_x+|z|\sigma_z \right)$,
with Bloch vector $\mathbf{r}_{\sigma_\rho}  = (c,0,|z|)$.
It follows immediately that $ |\mathbf{r}_{\sigma_\rho}|= |\mathbf{r}_{P_\rho}| = R$,
so $\sigma_\rho$ and $P_\rho$ have the same spectrum.

We next compute the coherent ergotropy. 
Using $H_0=\epsilon\ket{1}\bra{1}=\frac{\epsilon}{2}(I-\sigma_z),$
one has, for any qubit state with Bloch $z$ component $z_s$,
$\Tr(H_0\rho_s)= \frac{\epsilon}{2}(1-z_s)$.
The $z$ component is $|z|$ for $\sigma_\rho$, and $R$ for $P_\rho$.
Therefore,
\begin{align}
    \mathcal{E}_c(\rho)
    &=
    \Tr\!\left[
    H_0(\sigma_\rho-P_\rho)
    \right] \nonumber\\
    &=
    \frac{\epsilon}{2}
    \left[
    (1-|z|)-(1-R)
    \right] \nonumber\\
    &=
    \frac{\epsilon}{2}
    \left(
    R-|z|
    \right).
    \label{eq:app_Ec}
\end{align}
Equivalently, since $\Tr\rho^2 =\frac{1+R^2}{2}$,
we have $ 2\Tr\rho^2-1=R^2$.
Moreover, the $l_1$-coherence of $\rho$ in the energy basis is
$ C_{l_1}(\rho) = 2|\rho_{01}| = c$.
Hence $|z|=\sqrt{R^2-c^2}=\sqrt{2\Tr\rho^2-1-C_{l_1}^2(\rho)}$,
and Eq.~\eqref{eq:app_Ec} can also be written as
\begin{equation*}
    \mathcal{E}_c(\rho)
    =
    \frac{\epsilon}{2}
    \left[
    \sqrt{2\Tr\rho^2-1}
    -
    \sqrt{
    2\Tr\rho^2-1-C_{l_1}^2(\rho)
    }
    \right].
    \label{eq:app_Ec_coherence}
\end{equation*}

We now derive the coherent ergotropy extraction distance.
For two isospectral qubit states with Bloch vectors $R\mathbf{n}$ and $R\mathbf{m}$, $D(\rho,\sigma)= \frac{1}{2}\arccos(\mathbf{n}\cdot\mathbf{m})$.
In our case, $ \mathbf{n}_{\sigma_\rho}=\frac{1}{R}(c,0,|z|)$ and
$\mathbf{n}_{P_\rho}=(0,0,1)$.
Thus, $\mathbf{n}_{\sigma_\rho}\cdot\mathbf{n}_{P_\rho}=\frac{|z|}{R}$.
Therefore,
\begin{equation}
    D_c^{\rm ext}(\rho)
    =
    D(\sigma_\rho,P_\rho)
    =
    \frac{1}{2}
    \arccos\!\left(
    \frac{|z|}{R}
    \right).
    \label{eq:app_Dc}
\end{equation}
Finally, the geometric power capacity of coherent ergotropy is obtained by
\begin{align*}
\Pi_c(\rho)
=\frac{\mathcal{E}_c(\rho)} {D_c^{\rm ext}(\rho)}
=\frac{\frac{\epsilon}{2}(R-|z|)}{\frac{1}{2}\arccos(|z|/R) }
=\epsilon\frac{ R-|z|} {\arccos(|z|/R) }.\label{eq:app_Pi}
\end{align*}
For $c=0$, one has $R=|z|$, and therefore $\mathcal{E}_c(\rho)=0$ and $D_c^{\rm ext}(\rho)=0$.
In this incoherent limit, we define $\Pi_c(\rho)=0$ by continuity.
Indeed, for fixed $|z|>0$ and $c\rightarrow 0$, one has
$\frac{R-|z| }{\arccos(|z|/R) }\rightarrow 0$.

This completes the derivations of Eqs.~\eqref{eq:qubit_Ec_main}--\eqref{eq:qubit_Pi_main}.

\subsection{Derivations of the qutrit example}\label{app:qutrit_example}
We derive the expressions used in Sec.~\ref{subsec:qutrit_example}.
Consider the qutrit Hamiltonian,
\begin{equation*}
    H_0
    =
    \operatorname{diag}(\epsilon_1,\epsilon_2,\epsilon_3),
    ~~
    \epsilon_1<\epsilon_2<\epsilon_3,
\end{equation*}
and the state
\begin{equation}
    \rho
    =
    \begin{pmatrix}
        g_1 & c & 0\\
        c^* & g_3 & 0\\
        0 & 0 & g_2
    \end{pmatrix},
    ~~
    g_i=
    \frac{e^{-\beta\epsilon_i}}
    {\sum_j e^{-\beta\epsilon_j}},
    ~~
    |c|\leq\sqrt{g_1g_3}.
    \label{eq:app_qutrit_rho}
\end{equation}
For $\beta>0$, the  populations satisfy $ g_1\geq g_2\geq g_3$.

The state $\rho$ contains a single nontrivial coherent block,
\begin{equation*}
    B
    =
    \begin{pmatrix}
        g_1 & c\\
        c^* & g_3
    \end{pmatrix}.
\end{equation*}
The two eigenvalues of this block are
\begin{equation}
\lambda_\pm=\frac{g_1+g_3}{2} \pm \Delta,\label{eq:app_qutrit_lambda}
\end{equation}
where 
\begin{equation}
\Delta = \sqrt{\frac{(g_1-g_3)^2}{4} + |c|^2 } .\label{eq:app_qutrit_Delta}
\end{equation}
The third eigenvalue is $g_2$.

Since $\lambda_+\geq g_2\geq \lambda_-$,
the eigenvalues of $\rho$ in decreasing order are
$r_1=\lambda_+$, $r_2=g_2$ and $r_3=\lambda_-$.
Therefore, the passive state $P_\rho$ is obtained by assigning these eigenvalues to increasing energy levels,
\begin{equation*}
    P_\rho
    =
    \lambda_+
    \ket{\epsilon_1}\bra{\epsilon_1}
    +
    g_2
    \ket{\epsilon_2}\bra{\epsilon_2}
    +
    \lambda_-
    \ket{\epsilon_3}\bra{\epsilon_3}.
\end{equation*}
Equivalently, 
\begin{equation}
P_\rho =\operatorname{diag}(\lambda_+,g_2,\lambda_-).
\label{eq:app_qutrit_passive}
\end{equation}
The diagonal populations of $\rho$ are$ (g_1,g_3,g_2)$,
which are not in passive order.
The optimal incoherent unitary therefore exchanges the second and third energy levels, rearranging the populations into $ (g_1,g_2,g_3)$.
Under the same incoherent permutation, the coherence originally between levels $1$ and $2$ is moved to the subspace spanned by levels $1$ and $3$.
Thus, after extracting the incoherent ergotropy the state is
\begin{equation}
    \sigma_\rho
    =
    \begin{pmatrix}
        g_1 & 0 & c\\
        0 & g_2 & 0\\
        c^* & 0 & g_3
    \end{pmatrix}.
    \label{eq:app_qutrit_sigma}
\end{equation}
This state is unitarily connected to $\rho$, and hence has the same spectrum as $\rho$ and $P_\rho$.
The incoherent ergotropy extracted in the first process is
$\mathcal{E}_i(\rho) = \Tr[H_0(\rho-\sigma_\rho)]=(\epsilon_3-\epsilon_2)(g_2-g_3).$ 

The coherent ergotropy is the work extracted during the second process,
$\sigma_\rho \longrightarrow P_\rho$ ,
namely, $ \mathcal{E}_c(\rho)=\Tr[H_0(\sigma_\rho-P_\rho)].\label{eq:app_qutrit_Ec_def}$ 
Using Eqs.~\eqref{eq:app_qutrit_passive}and \eqref{eq:app_qutrit_sigma}, we obtain
\begin{align}
    \mathcal{E}_c(\rho)
    &=
    \epsilon_1(g_1-\lambda_+)
    +
    \epsilon_2(g_2-g_2)
    +
    \epsilon_3(g_3-\lambda_-)
    \nonumber\\
    &=
    \epsilon_1(g_1-\lambda_+)
    +
    \epsilon_3(g_3-\lambda_-).
    \label{eq:app_qutrit_Ec_step}
\end{align}
Since $ \lambda_+=\frac{g_1+g_3}{2}+\Delta$, $\lambda_-= \frac{g_1+g_3}{2}-\Delta $,
we have $ g_1-\lambda_+=\frac{g_1-g_3}{2}-\Delta $,
and $ g_3-\lambda_-=\Delta-\frac{g_1-g_3}{2} $.
Therefore,
\begin{align}
    \mathcal{E}_c(\rho)
    &=
    \epsilon_1
    \left(
    \frac{g_1-g_3}{2}-\Delta
    \right)
    +
    \epsilon_3
    \left(
    \Delta-\frac{g_1-g_3}{2}
    \right)
    \nonumber\\
    &=
    (\epsilon_3-\epsilon_1)
    \left[
    \Delta-\frac{g_1-g_3}{2}
    \right].
    \label{eq:app_qutrit_Ec}
\end{align}
This expression shows explicitly that $ \mathcal{E}_c(\rho)=0$ when $c=0$,
and that the coherent ergotropy increases with the magnitude of the coherent block.

We now compute $D_c^{\rm ext}(\rho) = D(\sigma_\rho,P_\rho)$. Although the full system is three-dimensional, the transformation from $\sigma_\rho$ to $P_\rho$ only acts nontrivially on the subspace $\operatorname{span}\{\ket{\epsilon_1},\ket{\epsilon_3}\}$.
The level $\ket{\epsilon_2}$ remains unchanged.
Write $ c=|c|e^{i\phi}$. The nontrivial $2\times2$ block of $\sigma_\rho$ is
\begin{equation}
    B_{13}
    =
    \begin{pmatrix}
        g_1 & c\\
        c^* & g_3
    \end{pmatrix}.
\end{equation}
It can be diagonalized by a unitary rotation in the $\{1,3\}$ subspace.
Let the mixing angle be $\theta$. Then $ \tan 2\theta =\frac{2|c|}{g_1-g_3}.\label{eq:app_qutrit_tan}$
Equivalently, $ \cos 2\theta=\frac{g_1-g_3}{2\Delta}.\label{eq:app_qutrit_cos}$
Thus,
\begin{equation}
    \theta
    =
    \frac{1}{2}
    \arccos
    \left(
    \frac{g_1-g_3}{2\Delta}
    \right).
    \label{eq:app_qutrit_theta}
\end{equation}
The optimal unitary connecting $\sigma_\rho$ to $P_\rho$ has nontrivial eigenvalues $e^{\pm i\theta}$ in the $\{1,3\}$ subspace and acts as the identity on $\ket{\epsilon_2}$.
Therefore, $  \|i\ln U_{\rm opt}\|=\theta$.
By the unitary expression of the quantum charging distance,
$ D_c^{\rm ext}(\rho) = D(\sigma_\rho,P_\rho) = \theta$.

Hence,
\begin{equation}
    D_c^{\rm ext}(\rho)
    =
    \frac{1}{2}
    \arccos
    \left(
    \frac{g_1-g_3}{2\Delta}
    \right).
    \label{eq:app_qutrit_Dc}
\end{equation}
Using Eq.~\eqref{eq:app_qutrit_Delta}, we also obtain
\begin{equation}
    D_c^{\rm ext}(\rho)
    =
    \frac{1}{2}
    \arccos
    \left(
    \frac{g_1-g_3}
    {
    \sqrt{(g_1-g_3)^2+4|c|^2}
    }
    \right).
    \label{eq:app_qutrit_Dc_explicit}
\end{equation}
The geometric power capacity of coherent ergotropy is defined as
$\Pi_c(\rho)=\frac{\mathcal{E}_c(\rho)}{D_c^{\rm ext}(\rho)}$.
Combining Eqs.~\eqref{eq:app_qutrit_Ec} and \eqref{eq:app_qutrit_Dc}, one obtains
\begin{equation}
    \Pi_c(\rho)
    =
    \frac{
    (\epsilon_3-\epsilon_1)
    \left[
    \Delta-\frac{g_1-g_3}{2}
    \right]
    }
    {
    \frac{1}{2}
    \arccos
    \left(
    \frac{g_1-g_3}{2\Delta}
    \right)
    }.
    \label{eq:app_qutrit_Pic}
\end{equation}
If the actual driving Hamiltonian has operator norm bounded by $\|V_t\|\leq \nu$, then the actual coherent discharging power satisfies
$P_c^{\rm ext}(\rho;V_t)\leq \nu\,\Pi_c(\rho)$.
Thus, $\Pi_c(\rho)$ gives the maximal coherent ergotropic power capacity per unit driving strength.

Define the normalized coherence parameter $ \eta =
\frac{|c|}{\sqrt{g_1g_3}}$, $0\leq\eta\leq 1.$ 
The positivity constraint $|c|\leq\sqrt{g_1g_3}$ is then automatically satisfied.
Equivalently, $ |c|=\eta\sqrt{g_1g_3}$.
Using the  populations
$g_i(\beta)= \frac{e^{-\beta\epsilon_i}} {\sum_j e^{-\beta\epsilon_j}}$,
we may write
\begin{equation}
    \Delta(\beta,\eta)
    =
    \sqrt{
    \frac{
    [g_1(\beta)-g_3(\beta)]^2
    }{4}
    +
    \eta^2 g_1(\beta)g_3(\beta)
    }.
\end{equation}
Therefore,
\begin{equation}
    \mathcal{E}_c(\beta,\eta)
    =
    (\epsilon_3-\epsilon_1)
    \left[
    \Delta(\beta,\eta)
    -
    \frac{g_1(\beta)-g_3(\beta)}{2}
    \right],
\end{equation}
\begin{equation}
    D_c^{\rm ext}(\beta,\eta)
    =
    \frac{1}{2}
    \arccos
    \left[
    \frac{
    g_1(\beta)-g_3(\beta)
    }
    {
    2\Delta(\beta,\eta)
    }
    \right],
\end{equation}
and
\begin{equation}
    \Pi_c(\beta,\eta)
    =
    \frac{
    \mathcal{E}_c(\beta,\eta)
    }
    {
    D_c^{\rm ext}(\beta,\eta)
    }.
\end{equation}
This parametrization makes the role of coherence transparent: $\eta=0$ corresponds to an incoherent state, while $\eta=1$ corresponds to the largest coherence allowed by positivity.

\subsection{Derivations of the three-level coherent discharging example}
\label{app:three_level_coherent_discharge}

We derive the expressions used in Sec.~\ref{subsec:three_level_coherent_discharge}.
The Hamiltonian is $ H_0=-h\ket{-h}\bra{-h}+ h\ket{h}\bra{h}$,
with basis ordering $ \{\ket{-h},\ket{0},\ket{h}\}$.
The passive Gibbs state is $$\rho_p =\frac{e^{-\beta H_0}}{Z}= \operatorname{diag}(a,b,d),$$
where $a=\frac{e^{\beta h}}{Z}$, $b=\frac{1}{Z}$, $d=\frac{e^{-\beta h}}{Z}$, and $Z=1+2\cosh(\beta h)$.
For $\beta>0$, one has $a\geq b\geq d$, so $\rho_p$ is passive.

We introduce the operator $ X=\ket{h}\bra{-h}+\ket{-h}\bra{h}$.
It acts nontrivially only in the subspace $\operatorname{span}\{\ket{-h},\ket{h}\}$
and satisfies $\|X\|=1 $.
The coherent family is defined by
\begin{equation}
    \rho_\theta
    =
    e^{-i\theta X}\rho_p e^{i\theta X}.
\end{equation}
Since $X$ acts as a Pauli operator on the subspace
$\mathcal{H}_{13}=\operatorname{span}\{\ket{-h},\ket{h}\}$ and vanishes on $\ket{0}$, one has
$$
X^2=I_{13},
\qquad
 I_{13}=\ket{-h}\bra{-h}+\ket{h}\bra{h}.
$$
Therefore, the unitary operator can be written as
$$ e^{-i\theta X} =
\ket{0}\bra{0} +\cos\theta\, I_{13} - i\sin\theta\, X .
 \label{eq:exp_theta_X}
$$
A direct calculation gives
\begin{equation}
    \rho_\theta
    =
    \begin{pmatrix}
        a\cos^2\theta+d\sin^2\theta & 0 & i(a-d)\sin\theta\cos\theta\\
        0 & b & 0\\
        -i(a-d)\sin\theta\cos\theta & 0 & a\sin^2\theta+d\cos^2\theta
    \end{pmatrix}.
    \label{eq:app_rho_theta_matrix}
\end{equation}
The energy-basis coherence is therefore
\begin{equation*}
    |\rho_{\theta,13}|
    =
    (a-d)\sin\theta\cos\theta
    =
    \frac{2\sinh(\beta h)}
    {1+2\cosh(\beta h)}
    \sin\theta\cos\theta .
    \label{eq:app_coherence_theta}
\end{equation*}

Since $\rho_\theta$ is unitarily connected to $\rho_p$, the two states have the same spectrum.
Because $\rho_p$ is passive, the passive state associated with $\rho_\theta$ is
$P_{\rho_\theta}=\rho_p$. Therefore, the extractable work from $\rho_\theta$ is precisely its total ergotropy,
$$
    \mathcal{E}(\rho_\theta)
    =
    \Tr[H_0(\rho_\theta-\rho_p)].
$$
Using Eq.~\eqref{eq:app_rho_theta_matrix}, we have the energy of $\rho_\theta$,
\begin{align*}
    \Tr(H_0\rho_\theta)
    &=
    -h\left(a\cos^2\theta+d\sin^2\theta\right)
    +
    h\left(a\sin^2\theta+d\cos^2\theta\right)
    \nonumber\\
    &=
    -h(a-d)\cos 2\theta.
\end{align*}
Meanwhile, $\Tr(H_0\rho_p) = -ha+hd= -h(a-d).$
Therefore,
\begin{align}
    \mathcal{E}(\rho_\theta)
    &=
    h(a-d)(1-\cos2\theta)
    2h(a-d)\sin^2\theta
    \nonumber\\
    &=
    \frac{4h\sinh(\beta h)}
    {1+2\cosh(\beta h)}
    \sin^2\theta.
    \label{eq:app_E_theta}
\end{align}
The dephased state of $\rho_\theta$ in the energy basis is
\begin{equation*}
    \Delta[\rho_\theta]
    =\operatorname{diag}
    \left(a\cos^2\theta+d\sin^2\theta,b,a\sin^2\theta+d\cos^2\theta\right).
\end{equation*}

For the total ergotropy $\mathcal{E}(\rho_\theta)$ to be purely coherent, the dephased state must remain passive.
This requires that $a\cos^2\theta+d\sin^2\theta\geq b\geq a\sin^2\theta+d\cos^2\theta$.
The second inequality is the more restrictive one, $a\sin^2\theta+d\cos^2\theta\leq b.$
Equivalently,
$ d+(a-d)\sin^2\theta\leq b$,
so that
\begin{equation*}
    \sin^2\theta
    \leq
    \frac{b-d}{a-d}.
\end{equation*}
Substituting $a=\frac{e^{\beta h}}{Z}$,
$b=\frac{1}{Z}$, $d=\frac{e^{-\beta h}}{Z},$
we obtain $\frac{b-d}{a-d}=\frac{1}{e^{\beta h}+1}$.
Thus, the coherent-only condition is
\begin{equation*}
    \sin^2\theta
    \leq
    \frac{1}{e^{\beta h}+1}.
\end{equation*}

Defining $ \theta_c(\beta) =\arcsin\left[\frac{1}{\sqrt{e^{\beta h}+1}}\right]$,
we have, for $0\leq \theta\leq \theta_c(\beta)$,
that the incoherent ergotropy vanishes $\mathcal{E}_i(\rho_\theta)=0$,
and hence $\mathcal{E}(\rho_\theta) =\mathcal{E}_c(\rho_\theta).$
Therefore, in this coherent-only region,
\begin{equation}
    \mathcal{E}_c(\rho_\theta)
    =
    \frac{4h\sinh(\beta h)}
    {1+2\cosh(\beta h)}
    \sin^2\theta.
    \label{eq:app_Ec_theta}
\end{equation}
In the coherent-only region, the incoherent ergotropy vanishes, and the active state for coherent extraction is simply $\sigma_{\rho_\theta} =\rho_\theta$.
The passive state is
$ P_{\rho_\theta} = \rho_p$.
Therefore, $D_c^{\rm ext}(\rho_\theta) = D(\rho_\theta,\rho_p)$.
Since $\rho_\theta=e^{-i\theta X}\rho_p e^{i\theta X}$ with $\|X\|=1$, the unit-norm driving Hamiltonian $X$ implements the transformation in time $\theta$.
The optimal unitary has nontrivial eigenvalues $e^{\pm i\theta}$, and hence
$\|i\ln U_{\rm opt}\|=\theta$. Thus, 
\begin{equation}
D_c^{\rm ext}(\rho_\theta) = \theta. 
 \label{eq:app_D_theta}
\end{equation}

Combining Eqs.~\eqref{eq:app_Ec_theta} and \eqref{eq:app_D_theta}, we obtain
\begin{equation}
    \Pi_c(\rho_\theta)
    =
    \frac{\mathcal{E}_c(\rho_\theta)}
    {D_c^{\rm ext}(\rho_\theta)}
    =
    \frac{4h\sinh(\beta h)}
    {1+2\cosh(\beta h)}
    \frac{\sin^2\theta}{\theta},
    \label{eq:app_Pi_theta}
\end{equation}
valid in the coherent-only region $0\leq \theta\leq \theta_c(\beta)$.
At $\theta=0$, the numerator and denominator both vanish, and we define
$\Pi_c(\rho_0)=0$ by continuity.
If the actual driving Hamiltonian has operator norm bounded by $\|V_t\|\leq \nu$, then the actual coherent discharging power satisfies
$ P_c^{\rm ext}(\rho_\theta;V_t) \leq
    \nu\,\Pi_c(\rho_\theta)$ .
The bound is saturated by the optimal driving Hamiltonian proportional to $X$.

\subsection{Illustrations of the general resource-geometric bounds}
\label{app:general_bound_examples}

\emph{Two-level system.}
We first illustrate the general resource-geometric bounds in
Eq.~\eqref{eq:Pi_c_bounds_abstract} using the two-level battery
introduced in Sec.~\ref{subsec:qubit_example}.
We use the same Hamiltonian $H_0=\epsilon\ket{1}\bra{1}$ and the same Bloch-state parametrization $\rho = \frac{1}{2}\left(I+x\sigma_x+y\sigma_y+z\sigma_z\right)$.
Let $ c=\sqrt{x^2+y^2}$ and $R=\sqrt{c^2+z^2}$.
For fixed Bloch radius $R$, we parametrize the population imbalance and the coherence amplitude as
$|z|=R\cos\alpha$, $c=R\sin\alpha$ for $0\leq \alpha\leq \frac{\pi}{2}$.
Here $\alpha$ measures the angular separation between the coherence-preserving active state
$\sigma_\rho$ and the passive state $P_\rho$ on the isospectral Bloch sphere.

For this qubit model, the coherent ergotropy, the coherent ergotropy extraction distance, and the geometric power capacity of coherent ergotropy are
\begin{equation}
    \mathcal{E}_c(\rho)
    =
    \frac{\epsilon R}{2}
    \left(
    1-\cos\alpha
    \right),
    \label{eq:qubit_Ec_exact_bounds}
\end{equation}
\begin{equation}
    D_c^{\rm ext}(\rho)
    =
    \frac{\alpha}{2},
    \label{eq:qubit_Dc_exact_bounds}
\end{equation}
respectively. Therefore,
\begin{equation}
    \Pi_c(\rho)
    =
    \frac{\mathcal{E}_c(\rho)}
    {D_c^{\rm ext}(\rho)}
    =
    \epsilon R
    \frac{1-\cos\alpha}{\alpha}.
    \label{eq:qubit_Pi_exact_bounds}
\end{equation}
At $\alpha=0$, we set $\Pi_c(\rho)=0$ by continuity.

We choose the auxiliary Gibbs state such that $ \rho_{\beta_*}=P_\rho $.
Equivalently,
\begin{equation*}
    \beta_*
    =    \frac{2}{\epsilon}
    \operatorname{arctanh} R,~~ 0<R<1.
    \label{eq:beta_star_qubit}
\end{equation*}
With this choice, the upper bound on the coherent ergotropy is saturated.
The distance bounds become
\begin{equation*}
    D_-(\rho)
    \leq
    D_c^{\rm ext}(\rho)
    \leq
    D_+,~~
    D_+ = \frac{\pi}{2},
    \label{eq:qubit_distance_bounds_example}
\end{equation*}
where $D_-(\rho)=\theta_B(\sigma_\rho,P_\rho)$ is the Bures-angle lower bound.

Combining the thermodynamic bounds on \(\mathcal{E}_c(\rho)\) with the geometric bounds on \(D_c^{\rm ext}(\rho)\) gives valid bounds on the geometric power capacity of coherent ergotropy,
\begin{equation}
    \frac{\mathcal{E}_c^-(\rho;\beta_*)}{D_+ }\leq\Pi_c(\rho)
    \leq\frac{\mathcal{E}_c^+(\rho;\beta_*)}{ D_-(\rho) }.
 \label{eq:qubit_Pi_bounds_window}
\end{equation}
Fig.~\ref{fig:qubit_bounds_example} illustrates this construction.
The left panel compares the exact coherent ergotropy with its thermodynamic lower and upper bounds.
The middle panel compares the exact coherent extraction distance with its geometric lower and upper bounds.
Combining these two sets of bounds yields the certified window for $\Pi_c(\rho)$ shown in the right panel.
This example shows that even when the coherent ergotropy can be computed exactly, the geometric power capacity of coherent ergotropy is further constrained by the minimal distance required to extract it.
\begin{figure*}[t]
    \centering
    \includegraphics[width=\linewidth]{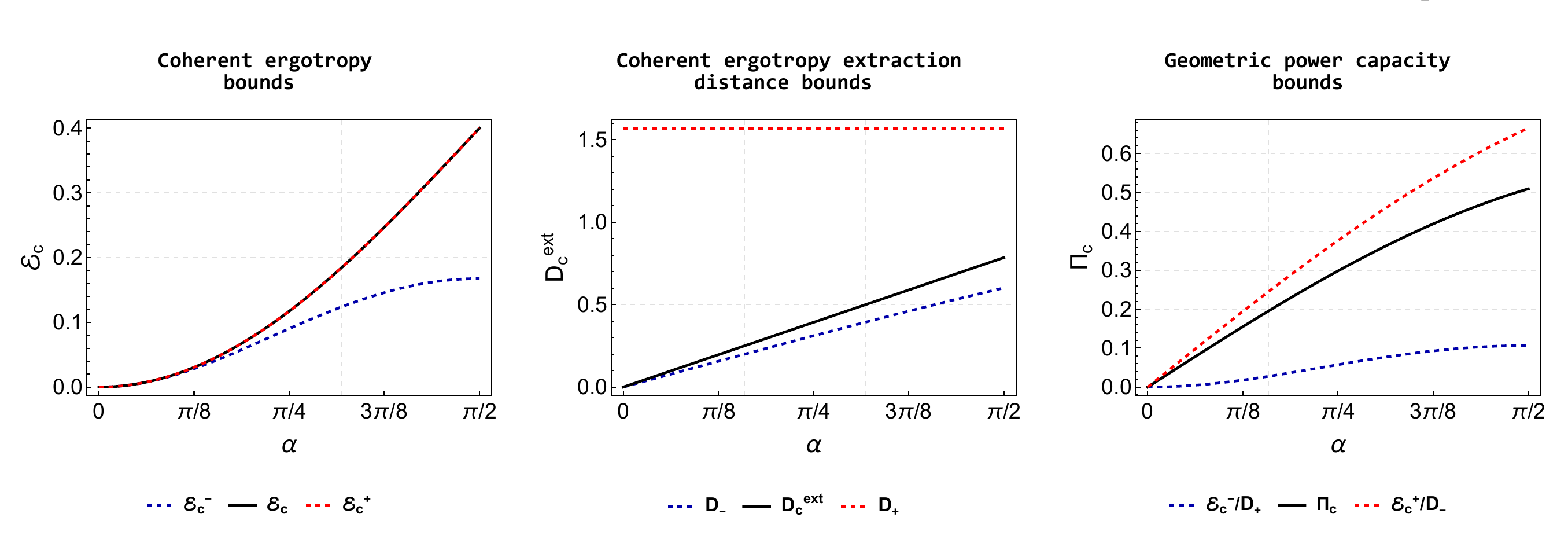}
    \caption{
    Qubit illustration of the bounds on the geometric power capacity  of coherent ergotropy.
    The left panel shows the exact coherent ergotropy $\mathcal{E}_c$ together with its bounds
    $\mathcal{E}_c^-$ and $\mathcal{E}_c^+$.
    The middle panel shows the exact coherent extraction distance $D_c^{\rm ext}$ together with the geometric bounds
    $D_-$ and $D_+=\pi/2$.
    The right panel shows the resulting certified bounds on
    $\Pi_c=\mathcal{E}_c/D_c^{\rm ext}$.
    The example demonstrates how thermodynamic bounds on coherent ergotropy and geometric bounds on extraction distance combine to constrain the geometric power capacity of coherent ergotropy.
    }
    \label{fig:qubit_bounds_example}
\end{figure*}

\emph{Three-level system.}\label{subsec:qutrit_bounds_example}
We next apply the same resource-geometric bounds to the qutrit family
introduced in Sec.~\ref{subsec:qutrit_example}. We take
$\varepsilon_1=0$, $\varepsilon_2=1$, $\varepsilon_3=2$ and use the normalized coherence strength
$\eta=\frac{|c|}{\sqrt{g_1g_3}}$, $0\leq \eta\leq 1 $. Here $\eta=0$ corresponds to the incoherent case, while $\eta=1$ saturates the positivity constraint of the coherent block. For this family, the exact coherent ergotropy $\mathcal{E}_c(\rho)$, coherent extraction distance $D_c^{\rm ext}(\rho)$ and geometric power capacity $\Pi_c(\rho)$ are given by Eqs.~\eqref{eq:qutrit_Ec}, \eqref{eq:qutrit_Dc} and \eqref{eq:qutrit_Pic}, respectively.
We compare these exact quantities with the general bounds derived.


The  $\mathcal{E}_c(\rho)$ is bounded as in Eq.~\eqref{eq:Ec_bounds_pm}, while the $D_c^{\rm ext}(\rho)$ satisfies the geometric bounds in Eq.~\eqref{eq:Dc_bounds_pm}.
For the  qutrit system ($d=3$), the dimension-dependent upper bound becomes
\begin{equation*}
    D_+
    =
    \pi\left(1-\frac{1}{3}\right)
    =
    \frac{2\pi}{3}.
\end{equation*}
The lower bound is given by $D_-(\rho)=\theta_B(\sigma_\rho,P_\rho)$.
Combining these two sets of bounds, for $D_-(\rho)>0$ we obtain the following two-sided bounds,
\begin{equation*}
\frac{\mathcal{E}_c^-(\rho;\beta) }{   D_+   }  \leq  \Pi_c(\rho)\leq\frac{ \mathcal{E}_c^+(\rho;\beta) }{ D_-(\rho)}.
\label{eq:qutrit_Pi_bounds_section4}
\end{equation*}

Thus, the geometric power capacity of coherent ergotropy is constrained simultaneously by the thermodynamic bounds on the coherent ergotropy and by the geometric bounds on the minimal extraction distance.
Fig.~\ref{fig:qutrit_bounds_example} shows the bound structure for this qutrit state.
As the normalized coherence strength $\eta$ increases, the coherent ergotropy increases, while the coherent extraction distance also grows because a larger unitary rotation is required to reach the passive state.
The right panel shows the resulting bounds on $\Pi_c(\rho)$.
This confirms that $\Pi_c(\rho)$ is not determined by the coherent ergotropy alone, but by the ratio between the coherent work content and the minimal geometric cost required to extract it.
\begin{figure*}[t]
    \centering
    \includegraphics[width=\linewidth]{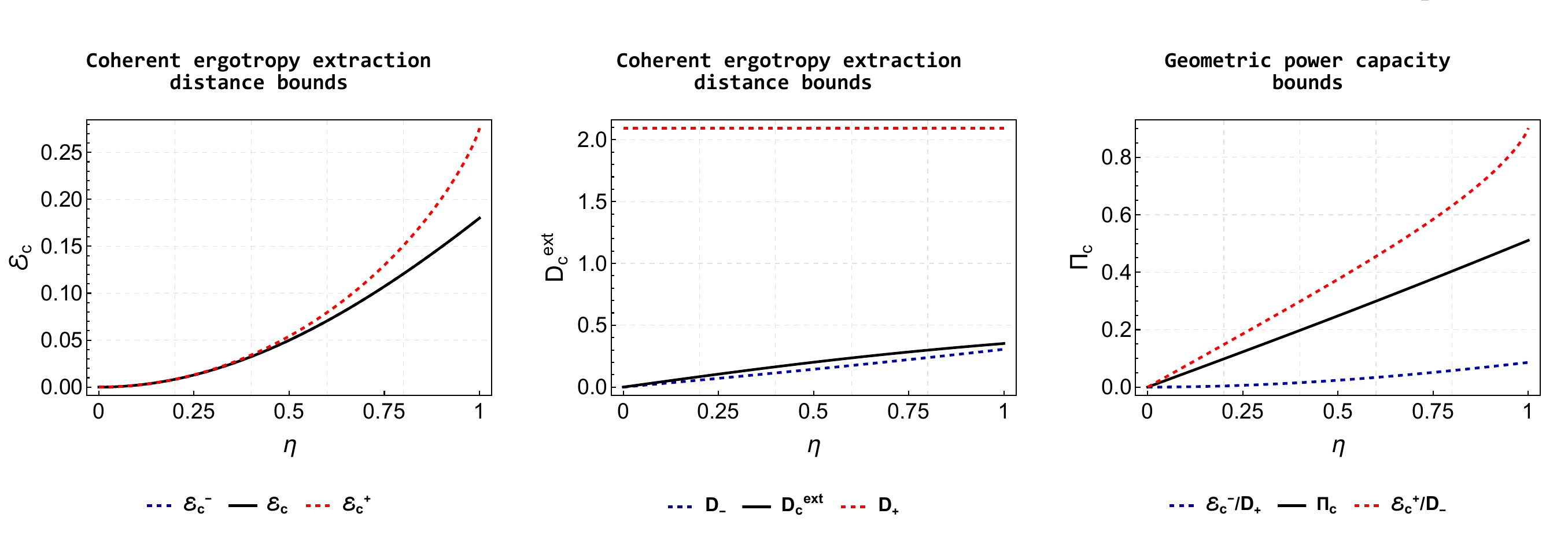}
    \caption{
    Bounds on the geometric power capacity of coherent ergotropy for qutrit case.
    The left panel compares the exact coherent ergotropy $\mathcal{E}_c$ with its thermodynamic lower and upper bounds $\mathcal{E}_c^-$ and $\mathcal{E}_c^+$ as functions of the normalized coherence strength $\eta=|c|/\sqrt{g_1g_3}$.
    The middle panel compares the exact coherent extraction distance $D_c^{\rm ext}$ with the geometric lower and upper bounds $D_-$ and $D_+=2\pi/3$. The right panel shows the resulting certified bounds on $\Pi_c=\mathcal{E}_c/D_c^{\rm ext}$. This example illustrates that the geometric power capacity of coherent ergotropy is bounded by the combined thermodynamic constraint on extractable coherent work and geometric constraint on the minimal extraction distance.}
    \label{fig:qutrit_bounds_example}
\end{figure*}

\subsection{Comparison of coherence-measure-induced bounds}
\label{Comparison of coherence-measure-induced bounds}
\emph{Qubit example}
\label{subsec:qubit_bounds_illustration}
We first illustrate the coherence-measure-induced bounds derived in
Sec.~\ref{sect5} using the qubit example discussed in Sec.~\ref{subsec:qubit_example}.
For a fixed Bloch radius $R$, we parametrize the state by
\begin{equation*}
    c=C_{l_1}(\rho)=R\sin\alpha,
~~
    |z|=R\cos\alpha,
~~
    0\leq \alpha\leq \frac{\pi}{2}.
    \label{eq:qubit_alpha_param_bounds}
\end{equation*}
In this case, we have
\begin{equation}
    \mathcal{E}_c(\rho)
    =
    \frac{\epsilon R}{2}\left(1-\cos\alpha\right),
    ~~
    D_c^{\rm ext}(\rho)
    =
    \frac{\alpha}{2},
    \label{eq:qubit_exact_Ec_Dc_bounds}
\end{equation}
and hence
\begin{equation}
    \Pi_c(\rho)
    =
    \epsilon R\,
    \frac{1-\cos\alpha}{\alpha},
    \label{eq:qubit_exact_Pi_bounds}
\end{equation}
with the value at $\alpha=0$ understood by continuity.

For the relative-entropy-induced bounds, we choose the Gibbs reference
$\rho_{\beta_*}$ such that
\begin{equation}
    P_\rho=\rho_{\beta_*},
~~    \beta_*=
    \frac{1}{\epsilon}
    \log\frac{1+R}{1-R}.
    \label{eq:beta_star_qubit_bounds}
\end{equation}
With this choice, the relative-entropy upper bound on coherent ergotropy is saturated,
$\mathcal{E}_{c,r}^+(\rho;\beta_*)=\mathcal{E}_c(\rho)$,
while the corresponding lower bound remains generally non-tight.
For the $l_1$-coherence measure, since
$C_{l_1}(\rho)=c$ for a qubit, positivity alone gives the simple bounds,
\begin{equation}
    \mathcal{E}_{c,l_1}^-(\rho)
    =
    \frac{\epsilon}{2}
    \left(
    1-\sqrt{1-C_{l_1}^2(\rho)}
    \right),~~
    \mathcal{E}_{c,l_1}^+(\rho)
    =
    \frac{\epsilon}{2}C_{l_1}(\rho).
    \label{eq:qubit_l1_direct_Ec_bounds}
\end{equation}

The energetic bounds and the geometric bounds on
\(D_c^{\rm ext}(\rho)\) together give the corresponding rigorous bounds on
\(\Pi_c(\rho)\). Fig.~\ref{fig:qubit_bounds_comparison} compares the exact
quantities with the bounds induced by different coherence measures.
The left panel shows the coherent ergotropy bounds, the middle panel shows the
resulting bounds on the geometric power capacity of coherent ergotropy, and the right panel compares the corresponding bound widths, $\Delta_M^\Pi=\Pi_M^+ - \Pi_M^-$.
A smaller value of \(\Delta_M^\Pi\) indicates a tighter estimation of the geometric power capacity of coherent ergotropy.
This example shows that different coherence quantifiers lead to different rigorous bounds on the same geometric power capacity. In the qubit case, the relative entropy construction is particularly sharp when the Gibbs reference is chosen as the passive state, since the upper energetic bound becomes exact. By contrast, the \(l_1\)-coherence bounds are more experimentally accessible and easier to compute, but generally produce wider bounds for
\(\Pi_c(\rho)\). The qubit example provides a simple benchmark for comparing
the tightness and practicality of different coherence-measure-induced bounds.
\begin{figure*}[t]
    \centering
    \includegraphics[width=1.05\linewidth]{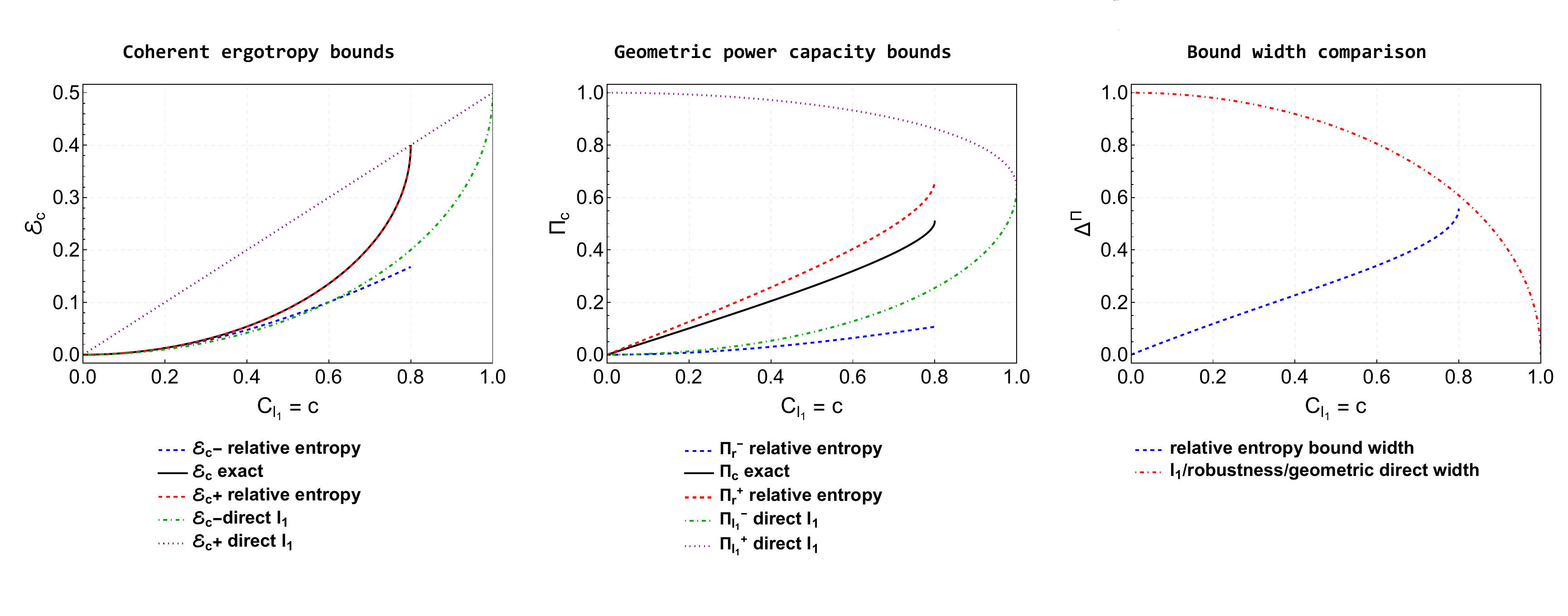}
    \caption{
    Comparison of coherence-measure-induced bounds for qubit systems.
    The state has fixed Bloch radius \(R\) and is parametrized by
    \(C_{l_1}(\rho)=c=R\sin\alpha\).
    Left: exact coherent ergotropy \(\mathcal{E}_c\) together with the
    relative-entropy-induced bounds and the \(l_1\)-coherence bounds.
    Middle: the corresponding bounds on the geometric power capacity of coherent ergotropy
    \(\Pi_c=\mathcal{E}_c/D_c^{\rm ext}\).
    Right: comparison of the bound widths
    \(\Delta_M^\Pi=\Pi_M^+-\Pi_M^-\) for different coherence estimates.
    For the Gibbs reference satisfying \(P_\rho=\rho_{\beta_*}\), the
    relative-entropy upper bound saturates the exact coherent ergotropy.
    The \(l_1\) bounds are simpler to evaluate but are generally less tight.}
    \label{fig:qubit_bounds_comparison}
\end{figure*}

\emph{Qutrit example}
\label{subsec:qutrit_bounds_illustration}
We next apply the same coherence-measure-induced bounds to the
single-coherent-block qutrit family introduced in Sec.~\ref{subsec:qutrit_example}.
Using the same notation as in Eqs.~(23)--(28), we regard the state as
parametrized by the inverse temperature \(\beta\) and the normalized
coherence strength \(\eta\).

This example is particularly useful for comparing different
coherence-measure-induced bounds, because the coherence is confined to a single
two-dimensional block. In this case the exact coherent ergotropy and coherent
extraction distance are known analytically, so they can be directly compared
with the bounds derived above.

Define
\begin{equation}
    \Delta(\beta,\eta)
    =
    \sqrt{
    \frac{[g_1(\beta)-g_3(\beta)]^2}{4}
    +
    \eta^2 g_1(\beta)g_3(\beta)
    } .
    \label{eq:qutrit_delta_beta_eta_bounds}
\end{equation}
The exact coherent ergotropy and coherent extraction distance are
\begin{equation}
    \mathcal{E}_c(\rho)
    =
    2\left[
    \Delta(\beta,\eta)
    -
    \frac{g_1(\beta)-g_3(\beta)}{2}
    \right]
    \label{eq:qutrit_exact_Ec_bounds_section}
\end{equation}
and
\begin{equation}
    D_c^{\rm ext}(\rho)
    =
    \frac{1}{2}
    \arccos
    \left[
    \frac{g_1(\beta)-g_3(\beta)}
    {2\Delta(\beta,\eta)}
    \right].
    \label{eq:qutrit_exact_Dc_bounds_section}
\end{equation}
Hence, $\Pi_c(\rho) = \frac{\mathcal{E}_c(\rho)}{D_c^{\rm ext}(\rho)}$.
    
In this case,
\begin{equation}
    C_{l_1}(\rho)=2|c|
    =
    2\eta\sqrt{g_1(\beta)g_3(\beta)} .
    \label{eq:qutrit_l1_eta_relation}
\end{equation}
Once \(\beta\) is fixed, the \(l_1\)-coherence already determines the
off-diagonal amplitude \(|c|\), and therefore determines the exact value of
\(\mathcal{E}_c(\rho)\). Explicitly,
\begin{align}
\mathcal{E}_{c,l_1}(\rho)
 &= 2\left[\sqrt{\frac{[g_1(\beta)-g_3(\beta)]^2}{4}+ \frac{C_{l_1}^2(\rho)}{4} } -\frac{g_1(\beta)-g_3(\beta)}{2}\right] \nonumber \\ 
 &= \mathcal{E}_c(\rho). \nonumber
\end{align}
Thus, for this special qutrit family, the \(l_1\)-coherence estimation is
just exact. Since the coherent extraction distance is also analytically known, the resulting
\(l_1\)-induced estimate of \(\Pi_c\) also coincides with the exact geometric
power capacity.

Fig.~\ref{fig:qutrit_bounds_comparison} compares these estimates as functions
of the normalized coherence strength \(\eta\). As \(\eta\) increases, both the
coherent ergotropy and the geometric power capacity of coherent ergotropy increase. The
relative-entropy bounds provide a general certified interval for
\(\Pi_c(\rho)\), but the interval is not tight for this particular family. By
contrast, the \(l_1\)-coherence estimate captures the exact coherent
ergotropy because the state contains only one coherent block and the thermal
populations are fixed by \(\beta\).
\begin{figure*}[t]
    \centering
    \includegraphics[width=1.0\linewidth]{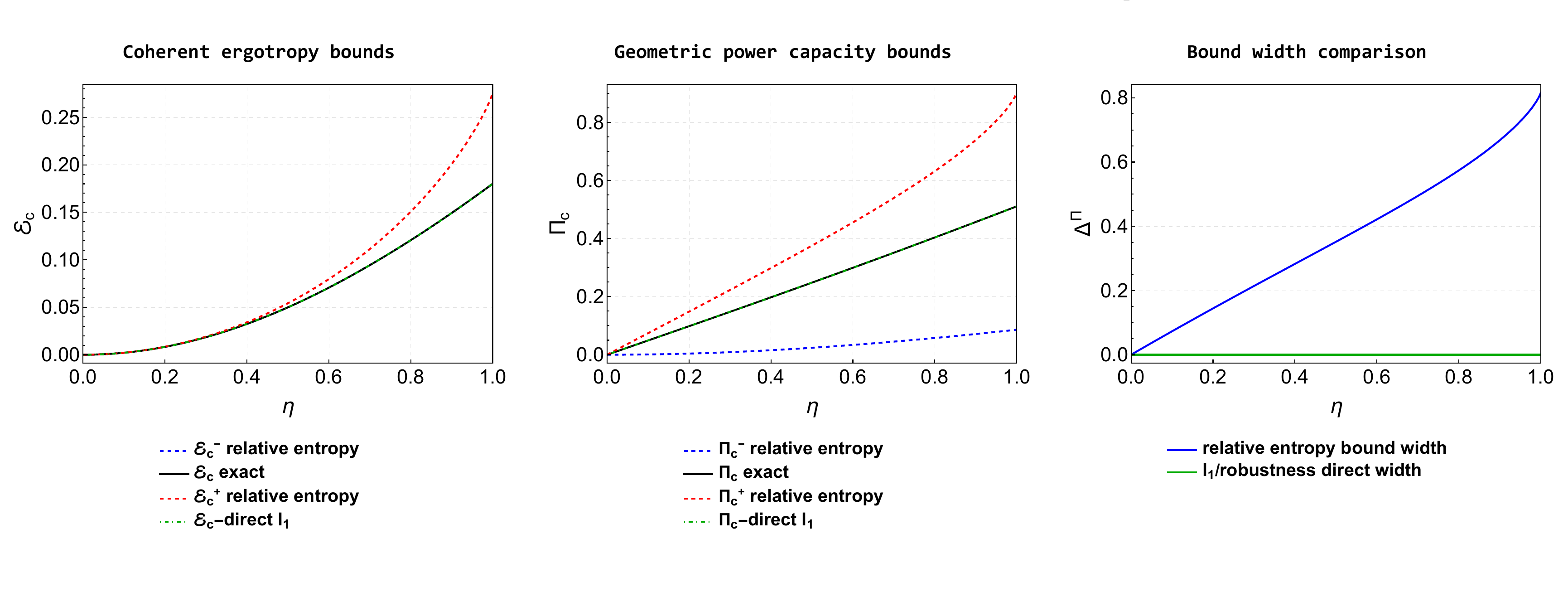}
    \caption{
    Comparison of coherence-measure-induced bounds for a single
    coherent block for the qutrit case. The state is parameterized by the normalized coherence
    strength \(\eta=|c|/\sqrt{g_1g_3}\), with fixed inverse temperature
    \(\beta\). Left: exact coherent ergotropy \(\mathcal{E}_c\) and its
    relative-entropy-induced bounds, together with the \(l_1\)-coherence
    estimate. Middle: the corresponding bounds on the geometric power capacity of coherent ergotropy
    \(\Pi_c=\mathcal{E}_c/D_c^{\rm ext}\). Right: comparison of the bound width
    \(\Delta^\Pi=\Pi^+-\Pi^-\). For this single-block qutrit family, the
    \(l_1\)-coherence estimate coincides with the exact result, while the
    relative-entropy construction gives a general but wider certified window.}
    \label{fig:qutrit_bounds_comparison}
\end{figure*}

This qutrit example therefore complements the qubit illustration. It shows that
the usefulness of a coherence-measure-induced bound depends on both the chosen
coherence measure and the structure of the state. For single-block models,
simple coherence amplitudes such as \(C_{l_1}\) can be highly effective, whereas
the relative-entropy construction remains more general but may lead to a wider
rigorous bound for the geometric power capacity of coherent ergotropy.


\bigskip
\begin{thebibliography}{57}


\bibitem{Alicki2013PRE}  R. Alicki and M. Fannes, Entanglement boost for extractable work from ensembles of quantum batteries, 
 \href{https://doi.org/10.1103/PhysRevE.87.042123}{\textit{Phys. Rev. E } \textbf{87}, 042123 (2013).}

\bibitem{Andolina2019PRB} G. M. Andolina, M. Keck, A. Mari, V. Giovannetti, and M.
Polini, Quantum versus classical many-body batteries,
\href{https://doi.org/10.1103/PhysRevB.99.205437}{\textit{Phys. Rev. B } \textbf{99},  205437 (2019).}

\bibitem{Andolina2019PRL} G. M. Andolina, M. Keck, A. Mari, M. Campisi, V. Giovannetti, and M. Polini, Extractable Work, the Role of Correlations, and Asymptotic Freedom in Quantum Batteries,
\href{https://doi.org/10.1103/PhysRevLett.122.047702}{\textit{Phys. Rev. Lett.} \textbf{ 122}, 047702 (2019).}

\bibitem{GyhmFischer2024} J.-Y. Gyhm and U. R. Fischer, Beneficial and detrimental entanglement for quantum battery charging,
\href{https://doi.org/10.1116/5.0184903}{\textit{AVS Quantum Sci.} \textbf{6}, 012001 (2024).}

\bibitem{Ferraro2018}D. Ferraro, M. Campisi, G. M. Andolina, V. Pellegrini, and M.
Polini, High-Power Collective Charging of a Solid-State Quantum Battery
\href{https://doi.org/10.1103/PhysRevLett.120.117702}{\textit{Phys. Rev. Lett.}, \textbf{120}, 117702 (2018) }

\bibitem{Binder2015} F. C. Binder, S. Vinjanampathy, K. Modi, and J. Goold,
Quantacell: powerful charging of quantum batteries,
\href{https://doi.org/10.1088/1367-2630/17/7/075015}{\textit{New J.Phys. } \textbf{17},075015  (2015).}


\bibitem{Le2018} T. P. Le, J. Levinsen, K. Modi, M. M. Parish, and F. A. Pollock,
Spin-chain model of a many-body quantum battery, \href{https://doi.org/10.1103/PhysRevA.97.022106}{\textit{Phys. Rev. A } \textbf{97},   022106 (2018).}


\bibitem{Andolina2018} G. M. Andolina, D. Farina, A. Mari, V. Pellegrini, V.
Giovannetti, and M. Polini, Charger-mediated energy transfer in exactly solvable models for quantum batteries,
\href{https://doi.org/10.1103/PhysRevB.98.205423}{\textit{Phys. Rev. B } \textbf{98},  205423 (2018).}

\bibitem{Ghosh2020} S. Ghosh, T. Chanda, and A. Sen(De),
Enhancement in the performance of a quantum battery by ordered and disordered interactions, \href{https://doi.org/10.1103/PhysRevA.101.032115}{\textit{ Phys. Rev. A } \textbf{101}, 032115 (2020).}

\bibitem{GarciaPintos2020} L. P. Garc\'{i}a-Pintos, A. Hamma, and A. del Campo, Fluctuations in Extractable Work Bound the Charging Power of Quantum Batteries,
\href{https://doi.org/10.1103/PhysRevLett.125.040601}{\textit{Phys. Rev. Lett.} \textbf{125}, 040601 (2020).}

\bibitem{Crescente2020} A. Crescente, M. Carrega, M. Sassetti, and D. Ferraro,
Charging and energy fluctuations of a driven quantum battery, 
\href{https://doi.org/10.1088/1367-2630/ab91fc}{\textit{New J. Phys. } \textbf{ 22},   063057 (2020).}

\bibitem{Gherardini2020} S. Gherardini, F. Campaioli, F. Caruso, and F. C. Binder, Stabilizing open quantum batteries by sequential measurements,
\href{https://doi.org/10.1103/PhysRevResearch.2.013095}{\textit{ Phy. Rev. Research} \textbf{2}, 013095 (2020).}


\bibitem{Santos2019} A. C. Santos, B. \c{C}akmak, S. Campbell, and N. T. Zinner, 
Stable adiabatic quantum batteries,
\href{https://doi.org/10.1103/PhysRevE.100.032107}{\textit{Phy. Rev. E } \textbf{100}, 032107 (2019).}


\bibitem{Quach2020} J. Q. Quach and W. J. Munro,
Using Dark States to Charge and Stabilize Open Quantum Batteries,
\href{https://doi.org/10.1103/PhysRevApplied.14.024092}{\textit{Phys. Rev. Appl. } \textbf{14}, 024092 (2020).}


\bibitem{Lu2025PRL}
Z.-G. Lu, G. Tian, X.-Y. L{\"u}, and C. Shang,
Topological quantum batteries,
\href{https://doi.org/10.1103/PhysRevLett.134.180401}{\textit{Phys. Rev. Lett. } \textbf{134}, 180401 (2025).}



\bibitem{ergotropyAllahverdyan2004Europhys}
A. E. Allahverdyan, R. Balian, and T. M. Nieuwenhuizen, Maximal work extraction from finite quantum systems,
\href{https://doi.org/10.1209/epl/i2004-10101-2}{\textit{Europhys. Lett.}\textbf{67}, 565 (2004).}





\bibitem{Farina2019}D. Farina, G. M. Andolina, A. Mari, M. Polini, and V.
Giovannetti, Charger-mediated energy transfer for quantum
batteries: An open-system approach, 
\href{https://doi.org/10.1103/PhysRevB.99.035421}{\textit{Phys. Rev. B } \textbf{99},  035421  (2019).}



\bibitem{Barra2019PRL} F. Barra, Dissipative charging of a quantum battery, 
\href{https://doi.org/10.1103/PhysRevLett.122.210601}{\textit{Phys. Rev. Lett.} \textbf{122}, 210601 (2019). }


\bibitem{Barra2022NJP} F. Barra, K. V. Hovhannisyan, and A. Imparato, Quantum
batteries at the verge of a phase transition,
\href{https://doi.org/10.1088/1367-2630/ac43ed}{\textit{New J. Phys.} \textbf{24}, 015003 (2022).}


\bibitem{SongWLPRL} W.-L. Song, H.-B. Liu, B. Zhou, W.-L. Yang, and J.-H. An,
Remote charging and degradation suppression for the
quantum battery, \href{https://doi.org/10.1103/PhysRevLett.}{\textit{Phys. Rev. Lett.} \textbf{ 132}, 090401 (2024).}

\bibitem{Campaioli2024RMP}F. Campaioli, S. Gherardini, J. Q. Quach, M. Polini, and
G. M. Andolina, Colloquium: Quantum batteries,
\href{https://doi.org/10.1103/RevModPhys.96.031001}{\textit{Rev. Mod. Phys.} \textbf{96}, 031001 (2024).}


\bibitem{Monsel2020PRL} J. Monsel, M. Fellous-Asiani, B. Huard, and A. Auff`eves,
The energetic cost of work extraction, 
\href{https://doi.org/10.1103/PhysRevLett.124.130601}{\textit{Phys. Rev. Lett.} \textbf{124}, 130601 (2020). }


\bibitem{Tirone2021PRL} S. Tirone, R. Salvia, and V. Giovannetti, Quantum energy
lines and the optimal output ergotropy problem, 
\href{https://doi.org/10.1103/PhysRevLett.127.210601}{\textit{Phys. Rev. Lett.} \textbf{127},  210601 (2021). }


\bibitem{GelbwaserKlimovsky2013} D. Gelbwaser-Klimovsky, R. Alicki, and G. Kurizki, Work and energy gain of heat-pumped quantized amplifiers,
\href{https://doi.org/10.1209/0295-5075/103/60005}{\textit{Europhys. Lett.}\textbf{103}, 60005 (2013). }


\bibitem{Niedenzu2016NJP} W. Niedenzu, D. Gelbwaser-Klimovsky, A. G. Kofman, and
G. Kurizki, On the operation of machines powered by
quantum non-thermal baths, 
\href{https://doi.org/10.1088/1367-2630/18/8/083012}{\textit{New J. Phys.}  \textbf{18}, 083012 (2016). }





\bibitem{Seah2018NJP}S. Seah, S. Nimmrichter, and V. Scarani, Work production
of quantum rotor engines,
\href{https://doi.org/10.1088/1367-2630/aab704}{\textit{New J. Phys.} \textbf{20}, 043045 (2018). }


\bibitem{Niedenzu2019Quantum} W. Niedenzu, M. Huber, and E. Boukobza, Concepts of
work in autonomous quantum heat engines,
\href{https://doi.org/10.22331/q-2019-10-14-195}{\textit{Quantum} \textbf{3}, 195 (2019). }

\bibitem{Seah2021PRL} S. Seah, M. Perarnau-Llobet, G. Haack, N. Brunner, and S.
Nimmrichter, Quantum speed-up in collisional battery charging, 
\href{https://doi.org/10.1103/PhysRevLett.127.100601}{\textit{Phys. Rev. Lett.} \textbf{127}, 100601 (2021).  }




\bibitem{coherentergotropyPRL}
G. Francica,  F. C. Binder,  G.Guarnieri, M. T.   Mitchison,   J. Goold, and F.Plastina, Quantum Coherence and Ergotropy,
\href{https://doi.org/10.1103/PhysRevLett.125.180603}{\textit{Phys. Rev. Lett.} \textbf{125},  180603 (2020).}

\bibitem{BPRE2020} B. \c{C}akmak, Ergotropy from coherences in an open quantum
system, \href{https://doi.org/10.1103/PhysRevE.102.042111}{\textit{Phy. Rev. E }\textbf{102}, 042111 (2020).}


\bibitem{Ju-Yeon2024PRA} Ju-Yeon Gyhm, Dario Rosa, and Dominik \ifmmode \check{S}\else \v{S}\fi{}afr\'anek,  Minimal time required to charge a quantum system,
\href{https://doi.org/10.1103/PhysRevA.109.022607}{\textit{Phys. Rev. A} \textbf{109},  022607 (2024).}


\bibitem{BCP2014}T. Baumgratz, M. Cramer, and M. B. Plenio, Quantifying
Coherence, \href{https://doi.org/10.1103/PhysRevLett.113.140401}{\textit{Phys. Rev. Lett.} \textbf{113},  140401 (2014).}

\bibitem{Rana2017PRA} S. Rana, P. Parashar, A. Winter, and M. Lewenstein, Logarithmic coherence: Operational interpretation of ${\ensuremath{\ell}}_{1}$-norm coherence,
\href{https://link.aps.org/doi/10.1103/PhysRevA.96.052336}{\textit{Phys. Rev. A} \textbf{96}, 052336 (2017).}


\bibitem{Napoli2016PRL}C. Napoli, T. R. Bromley, M. Cianciaruso, M. Piani, N. Johnston,
and G. Adesso, Robustness of Coherence: An Operational and Observable Measure of Quantum Coherence\href{https://doi.org/10.1103/PhysRevLett.116.150502}{\textit{Phys. Rev. Lett.} \textbf{116}, 116, 150502 (2016).}
\bibitem{Piani2016PRA} M. Piani, M. Cianciaruso, T. R. Bromley, C. Napoli, N. Johnston,
and G. Adesso, Robustness of asymmetry and coherence of quantum states ,
\href{https://doi.org/10.1103/PhysRevA.93.042107}{\textit{Phys. Rev. A} \textbf{93},  042107 (2016).}



\bibitem{Nielsenbook}
M. A. Nielsen and I. L. Chuang,\textit{Quantum Computation and Quantum Information}
(Cambridge University Press, Cambridge, England, 2000).

\bibitem{Streltsov2017RMP} Alexander Streltsov, Gerardo Adesso, and Martin B. Plenio, Colloquium: Quantum coherence as a resource,
\href{https://link.aps.org/doi/10.1103/RevModPhys.89.041003}{\textit{Rev. Mod. Phys.} \textbf{89}, 041003 (2017).}
\bibitem{Streltsov2015PRL} Alexander Streltsov, Uttam Singh, Himadri Shekhar Dhar, Manabendra Nath Bera, and Gerardo Adesso, Measuring Quantum Coherence with Entanglement, \href{https://link.aps.org/doi/10.1103/PhysRevLett.115.020403}{\textit{Phys. Rev. Lett.} \textbf{115}, 020403 (2015).}

\bibitem{MullerLennert2013Renyi} M. M\"{u}ller-Lennert, F. Dupuis, O. Szehr, S. Fehr, and M. Tomamichel, On quantum R\'{e}nyi entropies: A new generalization and some properties,
\href{https://doi.org/10.1063/1.4838856}{\textit{J. Math. Phys.} \textbf{54}, 122203 (2013).}
\bibitem{Campaioli2018PRL} F. Campaioli, F. A. Pollock, F. C. Binder, and K. Modi, Tightening Quantum Speed Limits for Almost All States,
\href{https://link.aps.org/doi/10.1103/PhysRevLett.120.060409}{\textit{Phys. Rev. Lett.} \textbf{120}, 060409 (2018).}
\bibitem{PhysRevX6021031} D. P. Pires, M. Cianciaruso, L. C. C\'{e}leri, G. Adesso and D. O. Soares-Pinto ,Generalized geometric quantum speed limits, \href{https://link.aps.org/doi/10.1103/PhysRevX.6.021031}{\textit{Phys. Rev. X} \textbf{6}, 021031 (2016).}
    
\bibitem{MT} L. Mandelstam and I. G. Tamm, The Uncertainty Relation Between Energy and Time in Non-relativistic Quantum Mechanics. In: Bolotovskii, B.M., Frenkel, V.Y., Peierls, R. (eds) SelectedPapers. \href{https://doi.org/10.1007/978-3-642-74626-0_8}{Springer, Berlin, Heidelberg, 1991.}
\bibitem{ML} N. Margolus and L. B. Levitin, The maximum speed of dynamical evolution, \href{https://www.sciencedirect.com/science/article/pii/S0167278998000542}{\textit{Phys. D} \textbf{120}, 188 (1998).}
    

    
\end{thebibliography}
\end{document}